\documentclass[12pt,reqno,a4paper]{amsart}
\usepackage[left=3cm, top=3.5cm, bottom=3.5cm, right=3cm]{geometry}
\usepackage{mathrsfs}
\usepackage{amsmath,amsxtra,amsfonts,amssymb, amsthm, amscd, epsfig}
\usepackage[dvipsnames,usenames]{color}
\usepackage{amssymb,amsmath,amsfonts,epsfig}
\usepackage{graphicx}
\usepackage[colorlinks=true, linkcolor=blue, citecolor=blue, urlcolor=blue]{hyperref}
\usepackage{tikz}
\usepackage{cite}

\numberwithin{equation}{section}

\newtheorem{thm}{Theorem}[section]

\newtheorem{cor}[thm]{Corollary}

\newtheorem{lem}[thm]{Lemma}

\newcommand{\eqa}{\begin{eqnarray}}
\newcommand{\eeqa}{\end{eqnarray}}
\newcommand{\beq}{\begin{equation}}
\newcommand{\eeq}{\end{equation}}
\newcommand{\nn}{\nonumber}
\newcommand{\p}{\partial}

 \def\res{\mathop{\text{\rm Res}}}

\def \dsum{\displaystyle\sum}

\allowdisplaybreaks[1]

\begin{document}
\title[Stability Limit of Prepotentials]{The Stability Limit of Prepotentials for Hurwitz--Frobenius Manifolds: An Infinite-Dimensional Approach}

\author{Shilin Ma}

\address{Shilin Ma, School of Mathematics and Statistics, Xi'an Jiaotong University, Xi'an 710049, P. R. China}
\email{mashilin@xjtu.edu.cn}

\date{\today}

\begin{abstract}
	The stability of prepotential derivatives for Frobenius manifolds associated with $A_N$ and $D_N$ singularities has been utilized to construct $(2+1)$-dimensional dispersionless integrable hierarchies. Although the generalization of this construction to genus-zero Hurwitz--Frobenius manifolds was shown to yield the genus-zero Whitham hierarchy, a direct geometric explanation of this correspondence has been lacking. In this note, we provide a direct proof of this identification within the framework of infinite-dimensional Frobenius manifolds. We demonstrate that the stability of prepotentials is an intrinsic property of the $\tau$-structure of the Whitham hierarchy. Furthermore, we extend this identification to the hierarchies arising from the stability of solutions to the open WDVV equations with the extensions of the Whitham hierarchy.
\end{abstract}

\maketitle

\section{Introduction}
The notion of Frobenius manifolds, introduced by B. Dubrovin \cite{dub1998}, provides a geometric framework for the Witten--Dijkgraaf--Verlinde--Verlinde (WDVV) equations in two-dimensional topological field theory (2D-TFT). In essence, a Frobenius manifold is characterized by a flat metric and a compatible, smoothly varying Frobenius algebra structure on its tangent bundle, with its quasi-homogeneity governed by a specific Euler vector field. Subsequently, Dubrovin and Zhang \cite{dubrovin2001normal} demonstrated that to any semisimple Frobenius manifold, one can associate a (1+1)-dimensional dispersionless, $\tau$-symmetric bi-Hamiltonian integrable system on its loop space, known as the principal hierarchy. The $\tau$-function of the topological solution to this hierarchy yields the genus-zero partition function of the corresponding TFT. Furthermore, the full-genus partition function is reconstructed via a quasi-Miura deformation of this hierarchy, which linearizes the Virasoro symmetries acting on the $\tau$-function. This theory generalizes the celebrated Witten conjecture \cite{witten1990two, kontsevich1992intersection}, which identifies the partition function of 2D topological quantum gravity with the $\tau$-function of the KdV hierarchy, and serves as a bridge between diverse disciplines, including Gromov-Witten theory \cite{manin1999frobenius, zhang2002cp1, carlet2004extended}, FJRW theory \cite{liu2015bcfg}, and the theory of Hodge integrals \cite{dubrovin2016hodge}. For recent developments and generalizations of the Dubrovin-Zhang theory, we refer the reader to \cite{liu2022variational, liu2023variational, liu2025generalized} and the references therein.

Beyond the Dubrovin-Zhang framework, several alternative methodologies have emerged to associate Frobenius manifolds with integrable hierarchies. For instance, in the context of cohomological field theory (CohFT), the double ramification (DR) hierarchies \cite{buryak2015double} are constructed through the intersection theory of double ramification cycles on the moduli space of curves $\overline{\mathcal{M}}_{g,n}$. For Frobenius manifolds arising from the Gromov–Witten theory of Fano orbifold curves, Hirota-type hierarchies have been established via Givental's quantization method \cite{milanov2007hirota, milanov2016gromov, cheng2021hirota}. This note focuses on the construction initiated by Basalaev, Dunin-Barkowski, and Natanzon \cite{basalaev2021integrable} and further developed in \cite{basalaev2024b, basalaev2025genus}. This framework derives a class of $ (2+1) $-dimensional dispersionless hierarchies by exploiting the stabilization of prepotential derivatives for several families of Frobenius manifolds, including those associated with $A_N$ and $D_N$ singularities \cite{dubrovin1998geometry}, and more generally, the genus-zero Hurwitz-Frobenius manifolds $H_{0;\mathbf{n}}$ \cite{dub1998}. In the latter case, these hierarchies were shown to be equivalent to the genus-zero universal Whitham hierarchy introduced by Krichever \cite{krichever1994tau}.

Basalaev and his co-authors identify these hierarchies by comparing them with the Hirota bilinear forms of known integrable systems. This approach provides limited insight into the geometric origin of the stabilization property and is not straightforwardly extensible to more general classes of Frobenius manifolds. To address these issues, the present note establishes a direct identification between the hierarchies arising from the stabilization of prepotentials and the $\tau$-covers of the principal hierarchies for certain infinite-dimensional Frobenius manifolds. Our approach utilizes the framework of infinite-dimensional Frobenius manifolds introduced by Carlet, Dubrovin, and Mertens \cite{carlet2011infinite}. This framework has since been used to characterize the geometric structures of several hierarchies, including the two-component B-type KP hierarchy \cite{wu2012class}, the Toda lattice hierarchy \cite{wu2014infinite}, and the extended KP hierarchy \cite{ma2021infinite}. Notably, the author, in collaboration with Wu and Zuo \cite{ma2024infinite}, has recently extended this approach to construct the infinite-dimensional Frobenius manifolds underlying the genus-zero universal Whitham hierarchy, thereby providing the necessary foundation for the current study.

More precisely, we demonstrate that the stable limits of prepotential derivatives for the Hurwitz-Frobenius manifolds $H_{0;\mathbf{n}}$ coincide with the $ \tau $-structure of a specific sector of the principal hierarchy for the infinite-dimensional Frobenius manifold associated with the universal Whitham hierarchy (Theorem \ref{preptau}). This result directly establishes the identification of the $ (2+1) $-dimensional hierarchies constructed by Basalaev \cite{basalaev2025genus} with the Whitham hierarchy (Corollary \ref{stabhier}). Furthermore, suitable reductions of this infinite-dimensional framework not only encompass the hierarchies derived from the stabilization of $A_N$ and $D_N$ type prepotentials \cite{basalaev2021integrable,basalaev2024b} as special cases, but also explicitly establish correspondence relations that remained largely implicit in the existing literature.

The framework established in \cite{basalaev2021integrable, basalaev2024b,basalaev2025genus} extends to integrable hierarchies associated with the stability of solutions to the open WDVV equations. These equations, introduced by Horev and Solomon \cite{horev2012open}, originate from the axioms of open Gromov-Witten theory and are applied in the computation of Welschinger invariants. Rossi showed that a solution to the open WDVV equations induces the structure of a flat $ F $-manifold on a one-dimensional extension of the underlying manifold. For solutions associated with $A_N$ and $D_N$ singularities \cite{basalaev2021open}, Basalaev \cite{basalaev2022integrable} constructed a hierarchy by utilizing the stabilization of prepotential derivatives and identified it with the dispersionless $\mathrm{mKP}$ hierarchy.

Recently, the author \cite{ma2025solutions} constructed solutions to the open WDVV equations for the infinite-dimensional Frobenius manifold underlying the genus-zero universal Whitham hierarchy and established the principal hierarchy for the associated flat $ F $-manifold. In the present note, we show that the $ \tau $-structure of the extended sector of this principal hierarchy coincides with the stable limits of the open WDVV derivatives for the genus-zero Hurwitz-Frobenius manifolds $H_{0;\mathbf{n}}$. This result provides an intrinsic geometric characterization of the stabilization property and identifies the hierarchy defined by these stable limits with the extended Whitham hierarchy (Corollary \ref{opencor}). Furthermore, the results of Basalaev \cite{basalaev2022integrable} are recovered as special cases of this identification.

The remainder of this note is organized as follows. In Section 2, we review the definitions of the principal hierarchy and its $\tau$-cover. Section 3 is devoted to a review of the infinite-dimensional Frobenius manifolds underlying the genus-zero universal Whitham hierarchy and the description of their $\tau$-covers. In Section 4, we provide the proof of our main result, which establishes a direct identification between the stable limits of the Hurwitz prepotential derivatives and the $\tau$-structure of a specific sector of the infinite-dimensional principal hierarchy ; we also extend this identification to the open sector. Finally, in Section 5, we investigate the reductions of this infinite-dimensional framework to even cases, which recovers and complements the known results for $D_N$ singularities.

\section{Preliminaries}
In this section, we review the essential definitions of Frobenius manifolds and their open extensions, as well as the associated principal hierarchies and $ \tau $-covers.
\subsection{Frobenius manifolds and their principal hierarchies}
A Frobenius manifold of charge $d$ is an $ n $-dimensional flat manifold $M$ such that each tangent space $T_v M$ is endowed with a Frobenius algebra structure $(T_v M, \circ, e, \langle \cdot, \cdot \rangle)$ that varies smoothly with $v \in M$. This structure is required to satisfy the following axioms:
\begin{enumerate}\item The bilinear form $\langle \cdot, \cdot \rangle$ provides a flat metric on $M$. The unit vector field $e$ is flat, i.e., $\nabla e = 0$, where $\nabla$ denotes the Levi-Civita connection of the flat metric.\item Define a 3-tensor $c$ by $c(X, Y, Z) := \langle X \circ Y, Z \rangle$ for any vector fields $X, Y, Z$ on $M$. Then the 4-tensor $(\nabla_W c)(X, Y, Z)$ is totally symmetric in $X, Y, Z, W$.\item There exists a vector field $E$, called the Euler vector field, which satisfies $\nabla^2 E = 0$ and the following conditions:
	$$[E, X \circ Y] - [E, X] \circ Y - X \circ [E, Y] = X \circ Y,$$
	$$E(\langle X, Y \rangle) - \langle [E, X], Y \rangle - \langle X, [E, Y] \rangle = (2 - d) \langle X, Y \rangle.$$
\end{enumerate}

Let $t = (t^1, \ldots, t^n)$ be a system of flat coordinates such that the unit vector field is given by $e = \frac{\partial}{\partial t^1}$. We denote by $\eta_{\alpha \beta}$ and $c^{\alpha}_{\beta \gamma}$ the components of the metric and the structure constants of the multiplication in these coordinates, respectively. The axioms of a Frobenius manifold ensure the local existence of a prepotential $ F(t) $, which satisfies
$$c_{\alpha \beta \gamma} = \frac{\partial^3 F}{\partial t^\alpha \partial t^\beta \partial t^\gamma},$$
$$\mathcal{L}_E F = (3-d) F + \text{quadratic terms in } t.$$
Consequently, the associativity of the Frobenius algebra is equivalent to the fact that $ F(t) $ satisfies the Witten-Dijkgraaf-Verlinde-Verlinde (WDVV) equations:
$$\frac{\partial^3 F}{\partial t^\alpha \partial t^\beta \partial t^\gamma} \eta^{\gamma \epsilon} \frac{\partial^3 F}{\partial t^\epsilon \partial t^\sigma \partial t^\mu} = \frac{\partial^3 F}{\partial t^\alpha \partial t^\sigma \partial t^\gamma} \eta^{\gamma \epsilon} \frac{\partial^3 F}{\partial t^\epsilon \partial t^\beta \partial t^\mu}.$$

The deformed flat connection on $ M $, originally introduced by Dubrovin in \cite{dub1998}, is defined for vector fields $X, Y \in \mathrm{Vect}(M)$ by:
$$\widetilde{\nabla}_X Y = \nabla_X Y + z X \circ Y, \quad X, Y \in \mathrm{Vect}(M).$$
This connection admits a consistent extension to a flat affine connection on $M \times \mathbf{C}^*$ through the following relations:
$$\begin{aligned}
	& \widetilde{\nabla}_X \frac{d}{dz} = 0, \\
	& \widetilde{\nabla}_{\frac{d}{dz}} \frac{d}{dz} = 0, \\
	& \widetilde{\nabla}_{\frac{d}{dz}} X = \partial_z X + E \circ X - \frac{1}{z} \mathcal{V}(X).
\end{aligned}$$
Here, $X$ is a vector field on $M \times \mathbf{C}^*$ whose $\mathbf{C}^*$ component vanishes, and the linear operator $\mathcal{V}$ is defined by
$$\mathcal{V}(X) := \frac{2-d}{2} X - \nabla_X E.$$

There exists a system of deformed flat coordinates $(\tilde{v}_1(t, z), \dots, \tilde{v}_n(t, z))$ that can be expressed in the following form:
\begin{equation}\label{defmcon}
(\tilde{v}_1(t, z), \dots, \tilde{v}_n(t, z)) = (\theta_1(t, z), \dots, \theta_n(t, z)) z^\mu z^R.
\end{equation}
These coordinates are chosen such that the 1-forms$$\xi_\alpha = \frac{\partial \tilde{v}_\alpha}{\partial t^\beta} dt^\beta, \quad \alpha = 1, \dots, n, \quad \text{and} \quad \xi_{n+1} = dz,$$constitute a basis of solutions to the system $\widetilde{\nabla} \xi = 0$. Here, $\mu = \mathrm{diag}(\mu_1, \dots, \mu_n)$ is a diagonal matrix determined by the spectrum of $M$:$$\mathcal{V} \left( \frac{\partial}{\partial t^\alpha} \right) = \mu_\alpha \frac{\partial}{\partial t^\alpha}, \quad \alpha = 1, \dots, n,$$and $R = R_1 + \dots + R_m$ is a constant nilpotent matrix satisfying the conditions:
$$\begin{aligned}
	& (R_s)_\beta^\alpha = 0 \quad \text{if} \quad \mu_\alpha - \mu_\beta \neq s, \\
	& (R_s)_\alpha^\gamma \eta_{\gamma \beta} = (-1)^{s+1} (R_s)_\beta^\gamma \eta_{\gamma \alpha}.
\end{aligned}$$

The functions $\theta_\alpha(t, z)$ appearing in the expression of the deformed flat coordinates are analytic at $z = 0$ and admit the power series expansion:
$$\theta_\alpha(t, z) = \sum_{p \ge 0} \theta_{\alpha, p}(t) z^p, \quad \alpha = 1, \dots, n.$$
The coefficients of this expansion satisfy the following recursive relations:
\begin{equation}
	\frac{\partial^2 \theta_{\alpha, p+1}(t)}{\partial t^\beta \partial t^\gamma} = c_{\beta \gamma}^\epsilon(t) \frac{\partial \theta_{\alpha, p}(t)}{\partial t^\epsilon}, \label{princon1}
\end{equation}
and
\begin{equation}
	\mathcal{L}_E \left( \frac{\partial \theta_{\alpha, p}(t)}{\partial t^\beta} \right) = (p + \mu_\alpha + \mu_\beta) \frac{\partial \theta_{\alpha, p}(t)}{\partial t^\beta} + \frac{\partial}{\partial t^\beta} \sum_{s=1}^p \theta_{\epsilon, p-s}(t) (R_s)_\alpha^\epsilon. \label{princon2}
\end{equation}
Furthermore, we impose the normalization condition:
\begin{equation}
	\theta_{\alpha, 0}(t) = \eta_{\alpha \beta} t^\beta. \label{princon3}
\end{equation}

The principal hierarchy associated with $M$ is defined as the Hamiltonian system on the loop space $\{ S^1 \to M \}$:
$$\frac{\partial t^\gamma}{\partial T^{\alpha, p}} = \{ t^\gamma(x), \int \theta_{\alpha, p+1}(t) dx \}_1 := \eta^{\gamma \beta} \frac{\partial}{\partial x} \left( \frac{\partial \theta_{\alpha, p+1}(t)}{\partial t^\beta} \right),$$
where $\alpha, \beta = 1, \dots, n$,  $p \ge 0$ , and $x = T^{1, 0}$ is the spatial variable. These flows are $ \tau $-symmetric, meaning that:
$$\frac{\partial \theta_{\alpha, p}(t)}{\partial T^{\beta, q}} = \frac{\partial \theta_{\beta, q}(t)}{\partial T^{\alpha, p}}=\partial_{x} \Omega_{\alpha, p; \beta, q}(t), \quad \forall \alpha, \beta, p, q.$$
The functions $\Omega_{\alpha, p; \beta, q}(t)$ are determined by the generating relation:
\begin{equation}
	\Omega_{\alpha, p-1; \beta, q} + \Omega_{\alpha, p; \beta, q-1} = \sum_{\mu, \nu} \eta^{\mu\nu} \frac{\partial \theta_{\alpha, p}}{\partial t^\mu} \frac{\partial \theta_{\beta, q}}{\partial t^\nu}, \quad p, q \ge 0,\label{taucru}
\end{equation}
subject to the symmetry property and boundary conditions:
 \begin{equation}\label{tausym}
 	\Omega_{\alpha, p; \beta, q} = \Omega_{\beta, q; \alpha, p}, \quad \Omega_{\alpha, p; 1, 0} = \theta_{\alpha, p},\quad \Omega_{\alpha, 0; \beta, q} = \frac{\partial \theta_{\beta, q+1}}{\partial t^{\alpha}}.
 \end{equation}
One can then verify that
$$\frac{\partial \Omega_{\beta, q; \gamma, r}}{\partial T^{\alpha, p}} = \frac{\partial \Omega_{\alpha, p; \gamma, r}}{\partial T^{\beta, q}}$$
and
$$\Omega_{\alpha, 0; \beta, 0}(t) = \frac{\partial^2 F(t)}{\partial t^\alpha \partial t^\beta}.$$
Since $t^{\alpha} = \eta^{\alpha \mu} \theta_{\mu, 0}$, the evolution equations of the principal hierarchy can be rewritten as:
$$\frac{\partial t^\gamma}{\partial T^{\alpha, p}} = \eta^{\gamma \mu} \partial_{x} \Omega_{\alpha, p; \mu, 0}(t).$$
The $ \tau $-cover of the principal hierarchy is given by the following system of second-order logarithmic equations:
$$\frac{\partial^2 \log \tau}{\partial T^{\alpha, p} \partial T^{\beta, q}} = \Omega_{\alpha, p; \beta, q}(t)$$
where the coordinates $t^\alpha$ are expressed in terms of the $ \tau $-function as
$$t^{\alpha} = \eta^{\alpha \mu} \frac{\partial^2 \log \tau}{\partial T^{\mu, 0} \partial x}.$$
Since any solution to the $ \tau $-cover equations yields a solution to the principal hierarchy, this construction effectively lifts the hierarchy from the manifold coordinates to the space of the $ \tau $-function (see \cite{dubrovin2001normal} for further details).

\subsection{Open WDVV equations and flat $F$-manifolds}
A flat $F$-manifold $(M, \nabla, \circ, e)$ consists of an analytic manifold $M$, a flat torsionless connection $\nabla$ on the tangent bundle $TM$, and a commutative, associative algebra structure on each tangent space $T_p M$ with a unit vector field $ e $, satisfying the following conditions:
\begin{enumerate}
	\item[(1)] $\nabla e=0$;
	\item[(2)] There exists a vector field $\Psi$ on $M$, called the vector potential, such that
	$$
	X\circ Y=[X,[Y,\Psi]]
	$$
	for any flat vector field $X$ and $Y$ on $M$.
\end{enumerate}
For further details and properties of flat $ F $-manifolds, we refer the reader to \cite{manin2005f, arsie2018flat, alcolado2017extended}.

As observed by P. Rossi, a solution $F^o(t, s)$ to the open WDVV equations:
$$
c_{\alpha\beta}^{\delta}\frac{\p  F^{o}(t,s)}{\p t^{\delta}\p t^{\gamma}}+\frac{\p  F^{o}(t,s)}{\p t^{\alpha}\p t^{\beta}}\frac{\p  F^{o}(t,s)}{\p t^{\gamma}\p s}=c_{\beta\gamma}^{\delta}\frac{\p  F^{o}(t,s)}{\p t^{\delta}\p t^{\alpha}}+\frac{\p  F^{o}(t,s)}{\p t^{\beta} \p t^{\gamma}}\frac{\p  F^{o}(t,s)}{\p t^{\alpha}\p s}
$$
and
$$
c_{\alpha\beta}^{\delta}\frac{\p  F^{o}(t,s)}{\p  t^{\delta} \p s}+\frac{\p  F^{o}(t,s)}{\p t^{\alpha} \p t^{\beta}}\frac{\p  F^{o}(t,s)}{\p s\p s}=\frac{\p  F^{o}(t,s)}{\p t^{\alpha} \p s}\frac{\p  F^{o}(t,s)}{\p t^{\beta}\p s},
$$
define a flat F-manifold 
$(\tilde{M}, \nabla, *,e)$ where
$\tilde{M} = M \times \mathbf{C}$ with $t^1, \dots, t^n$ be the flat coordinates on $M$ and $s$ be the coordinate on $\mathbf{C}$. 
The flat connection $\nabla$ on $\tilde{M}$ is determined via the flat coordinates $\{t^1, \ldots, t^n, s\}$. For the basis vectors $\partial_A$ and $\partial_B$ with indices $A, B \in \{1, \ldots, n, s\}$, the product is defined as
$$\partial_A * \partial_B = \sum_{C} c_{AB}^C \partial_C$$
where the structure constants $c_{AB}^C$ are specified as follows. The components $c_{\alpha\beta}^\gamma$ for $\alpha, \beta, \gamma \in \{1, \ldots, n\}$ coincide with those of the underlying Frobenius manifold $M$, while the components involving the open sector are given by
$$c_{\alpha \beta}^s = \frac{\partial^2 F^o(t, s)}{\partial t^\alpha \partial t^\beta}, \quad c_{\alpha s}^s=\frac{\partial^2 F^o(t, s)}{\partial t^\alpha \partial s},\quad  c_{1 s}^s=1,\quad c_{\alpha s}^\beta= c_{s s}^\beta= 0.$$

Let $\Theta_{A, p} = \sum_{B} \Theta_{A, p}^B \partial_B$ be a sequence of vector fields on $\tilde{M}$ satisfying 
\begin{equation}\label{openrecu}
	\nabla_X \Theta_{A, p+1} = \Theta_{A, p} * X,\quad \Theta_{A,0}=\p_{A}.
\end{equation}
The principal hierarchy \cite{arsie2018flat} associated with the flat $ F $-manifold $\tilde{M}$ is defined by the system:
$$\frac{\partial \mathbf{w}}{\partial T^{A, p}} = \Theta_{A, p} * \frac{\partial \mathbf{w}}{\partial x}$$
where $\mathbf{w} = (t^1, \ldots, t^n, s)^T$.
We assume that the components of these flows on the Frobenius manifold sector $M$ coincide with the principal hierarchy for $ M $, namely:
$$\Theta_{\alpha, p}^\beta = \sum \eta^{\beta\gamma} \partial_\gamma \theta_{\alpha, p}(t),\quad \Theta_{s, p}^\beta=0,$$
for $\alpha, \beta, \gamma \in \{1, \ldots, n\}$, and denote $\tilde{\theta}_{A, p} = \Theta_{A, p+1}^s.$
The recursion relation then implies that the functions $\tilde{\theta}_{\alpha, p}$ satisfy the following system of equations \cite{basalaev2019open}:
	$$\frac{\partial \tilde{\theta}_{\alpha, p}}{\partial s} = \sum_{\mu, \nu=1}^n \eta^{\mu\nu} \frac{\partial \theta_{\alpha, p}}{\partial t^\nu} \frac{\partial^2 F^o}{\partial t^\mu \partial s} + \tilde{\theta}_{\alpha, p-1} \frac{\partial^2 F^o}{\partial s^2},$$
$$\frac{\partial \tilde{\theta}_{\alpha, p}}{\partial t^\beta} = \sum_{\mu, \nu=1}^n \eta^{\mu\nu} \frac{\partial \theta_{\alpha, p}}{\partial t^\nu} \frac{\partial^2 F^o}{\partial t^\mu \partial t^\beta} + \tilde{\theta}_{\alpha, p-1} \frac{\partial^2 F^o}{\partial s \partial t^\beta},$$
$$\frac{\partial \tilde{\theta}_{s, p}}{\partial s} = \tilde{\theta}_{s, p-1} \frac{\partial^2 F^o}{\partial s^2},\quad \frac{\partial \tilde{\theta}_{s, p}}{\partial t^\beta} = \tilde{\theta}_{s, p-1} \frac{\partial^2 F^o}{\partial s \partial t^\beta}$$
for $p\ge -1$, with the initial condition $\tilde{\theta}_{\alpha, -1} = 0$. In particular, the low-order terms of the recursion can be solved as follows:
$$\tilde{\theta}_{\alpha, 0} = \frac{\partial F^o}{\partial t^\alpha}, \quad \tilde{\theta}_{s, p} = \frac{1}{(p+1)!} \left( \frac{\partial F^o}{\partial s} \right)^{p+1}.$$
Then the principal hierarchy can be expressed as 
\begin{align*}
	&\frac{\partial t^\mu (x)}{\partial T^{\beta, q}} = \frac{\partial}{\partial x} \left( \sum_{\nu=1}^n \eta^{\mu\nu} \frac{\partial \theta_{\beta, q+1}(t)}{\partial t^\nu} \right), \quad \frac{\partial s(x)}{\partial T^{\beta, q}} = \frac{\partial \tilde{\theta}_{\beta, q}(t,s)}{\partial x},\\
	&\frac{\partial t^\mu (x)}{\partial T^{s, q}} = 0, \quad \frac{\partial s(x)}{\partial T^{s, q}} = \frac{\partial \tilde{\theta}_{s, q}(t,s)}{\partial x}.
\end{align*}
By identifying $s = \partial_1 F^o = \tilde{\theta}_{1, 0}$ and utilizing the symmetry relations
$$\frac{\partial \tilde{\theta}_{A, p}}{\partial T^{B, q}} = \frac{\partial \tilde{\theta}_{B, q}}{\partial T^{A, p}},$$
the open $\tau$-cover for the principal hierarchy is defined following \cite{basalaev2019open}.  This construction introduces two $ \tau $-functions, $\tau_1$ and $\tau_2$, which satisfy the following system of logarithmic equations:
$$\frac{\partial^2 \log \tau_{1}}{\partial T^{\alpha, p} \partial T^{\beta, q}} = \Omega_{\alpha, p; \beta, q}(t),\quad \frac{\partial \log \tau_{1}}{\p T^{s,p}}=0,\quad \frac{\p \log \tau_{2}}{\p T^{A, p}}=\tilde{\theta}_{A,p}(t,s)$$
for $p\ge 0$, where
$$t^{\alpha} = \eta^{\alpha \mu} \frac{\partial^2 \log \tau_{1}}{\partial T^{\mu, 0} \partial x},\quad s=\frac{\p \log \tau_{2}}{\p x}$$
with $x=T^{1,0}.$

\section{$\tau$-structure of $\mathcal{M}$}
In this section, we review the infinite-dimensional Frobenius manifold $\mathcal{M}$ associated with the genus-zero universal Whitham hierarchy and its corresponding principal hierarchy. As established in \cite{ma2024infinite}, the genus-zero universal Whitham hierarchy is recovered by restricting the flows of the principal hierarchy to a specific sector. Building upon this framework, we then derive an explicit expression for the $\tau$-structure of the principal hierarchy within this sector.

\subsection{Definitions of $\mathcal{M}$}
We denote by $\mathcal{M}$ the space consisting of pairs of functions $\vec{a} = (a(z), \hat{a}(z))$. These functions are analytic in the domains $\mathbf{D}^{\text{ext}} = \mathbf{P}^1 \setminus \bigcup_{j=1}^m D_j$ and $\mathbf{D}^{\text{int}} = \bigcup_{j=1}^m D_j$ respectively, where $D_1, \ldots, D_m$ denote a set of pairwise disjoint closed disks. Specifically, $a(z)$ is required to have a simple pole at $z = \infty$, while $\hat{a}(z)$ possesses simple poles at the points $\varphi_j \in D_j$ for each $j = 1, \ldots, m$. Their Laurent expansions in the neighborhood of these poles are given by:
\begin{align}
	&a(z)=z+a_{1}z^{-1}+a_{2}z^{-2}+\cdots,\quad z\to \infty;\nn \\	
	&\hat{a}(z)=\hat{a}_{j,-1}(z-\varphi_{j})^{-1}+\hat{a}_{j,0} +\hat{a}_{j,1} (z-\varphi_{j})+\cdots, \quad z\to \varphi_{j}.\label{formal}
\end{align}
Furthermore, we assume that $a(z)$ and $\hat{a}(z)$ have no zeros in their respective domains $\mathbf{D}^{\text{ext}}$ and $\mathbf{D}^{\text{int}}$.

For any $\vec{a} = (a(z), \hat{a}(z)) \in \mathcal{M}$, we identify each vector $X$ in the tangent space $T_{\vec{a}} \mathcal{M}$ with a pair of directional derivatives $(\partial_X a(z), \partial_X \hat{a}(z))$. The metric on $T_{\vec{a}}\mathcal{M}$ is defined by
\begin{equation}\label{whimetric}
	\langle X_1,X_2\rangle_{\eta}=
	-\frac{1}{2\pi\mathrm{i}}\dsum_{j=1}^{m}\oint_{\gamma_{j}}\frac{\p_{1}\zeta(z) \cdot\p_{2}\zeta(z)}{\zeta'(z)}dz
	-\left(\res_{z=\infty}+\dsum_{j=1}^{m}\res_{z=\varphi_{j}}\right)\frac{\p_{1}\ell(z)\cdot \p_{2}\ell(z)}{\ell'(z)}dz,
\end{equation}
 for $X_1, X_2 \in T_{\vec{a}} \mathcal{M}$, where $\partial_{\nu} = \partial_{X_{\nu}}$ for $\nu = 1, 2$. The auxiliary functions $\zeta(z)$ and $\ell(z)$ are defined as
\begin{equation}
	\zeta(z)=a(z)-\hat{a}(z),\quad \ell(z)=a(z)_{+}+\hat{a}(z)_{-}.
\end{equation}
Here, the subscripts $+$ and $-$ denote the components of the unique decomposition $f(z) = f(z)_+ + f(z)_-$. This decomposition is applicable to any function $f(z)$ that is holomorphic in a neighborhood of the boundary $\gamma = \bigcup_{j=1}^m \gamma_j$, with $\gamma_j = \partial D_j$. Specifically, $f(z)_+$ is analytic in $\mathbf{D}^{\mathrm{int}}$, and $f(z)_-$ is analytic in $\mathbf{D}^{\mathrm{ext}}$ and satisfies $f_-(\infty) = 0$.

Let $\mathcal{H}$ denote the space of functions holomorphic in a neighborhood of the domain boundaries $\gamma = \bigcup_{j=1}^m \gamma_j$. The metric can be dually characterized via a linear map $\eta: \mathcal{H} \times \mathcal{H} \to T_{\vec{a}} \mathcal{M}$ defined as
\begin{align}
	\eta(\omega(z), \hat{\omega}(z))=&\bigl(a'(z)(\omega(z)+\hat{\omega}(z))_{-}-(\omega(z) a'(z)+\hat{\omega}(z)\hat{a}'(z))_{-},\nn\\
	&-\hat{a}'(z)(\omega(z)+\hat{\omega}(z))_{+}+(\omega(z) a'(z)+\hat{\omega}(z)\hat{a}'(z))_{+}\bigr). \label{etaom}
\end{align}
This map satisfies the identity
\begin{equation}\label{pairmec}
	\langle\vec{\omega},X\rangle=\langle\eta (\vec{\omega}),X\rangle_{\eta},
\end{equation}
where the pairing between $\vec{\omega} = (\omega(z), \hat{\omega}(z))\in\mathcal{H}\times\mathcal{H}$ and a tangent vector $X \in T_{\vec{a}} \mathcal{M}$ is defined by the following contour integral:
\begin{equation}\label{whipair}
	\langle \vec{\omega}, X \rangle := \frac{1}{2\pi \mathrm{i}} \sum_{j=1}^{m} \oint_{\gamma_j} \left( \omega(z) \partial_X a(z) + \hat{\omega}(z) \partial_X \hat{a}(z) \right) \, dz.
\end{equation}
As established in \cite{ma2024infinite}, the map $\eta$ is surjective. This property allows for a natural identification of the cotangent space $T_{\vec{a}}^\ast \mathcal{M}$ with the following quotient space:$$T_{\vec{a}}^\ast \mathcal{M} = (\mathcal{H} \times \mathcal{H}) / \ker \eta.$$

The flat coordinates of the metric $\eta$ are constructed as follows. 
We consider the restriction of the auxiliary function $\zeta(z)$ to the boundaries $\gamma_j$ and denote
\[
\zeta(z)|_{\gamma_j} = w_j(z)^{d_j}|_{\gamma_j}, \quad w'(z)|_{\gamma_j} \neq 0,
\]
where each $w_j(z)$ is a conformal map that sends $\gamma_j$ to a closed path $\tilde{\gamma}_j$ in the $w$-plane, encircling the origin with winding number $1$. The inverse function $z(w_j)$ admits a Laurent expansion of the form
\begin{equation}\label{texp}
	z(w_{j})=\sum_{l \in \mathbb{Z}} t_{j,l} w_{j}^l, \quad w_{j} \in \tilde{\gamma}_j,
\end{equation}
where the coefficients $t_{j, l}$ are determined by the following contour integrals:
\[
t_{j,l}=\frac{1}{2 \pi \mathrm{i}} \oint_{\tilde{\gamma}_j} z(w_{j}) w_{j}^{-l-1} d w_{j}, \quad  j\in\{1,2,\cdots,m\},\ l \in \mathbb{Z}.
\]
Near the poles $\varphi_j$, the auxiliary function $\ell(z)$ exhibits the following behavior:
$$
\begin{aligned}
	& \ell(z)=b_{j,-1}(z-\varphi_{j})^{-1}+b_{j,0}+\mathcal{O}(z-\varphi_{j}), \quad z \rightarrow \varphi_{j} .
\end{aligned}
$$
The inverse function $z(\ell)$ can then be expanded into a Laurent series as follows:
\begin{align}
	& z(\ell)=h_{j,0}+h_{j,1} \ell^{-1}+\mathcal{O}(\ell^{-2}), \quad z \rightarrow \varphi_{j} .\label{hhexp}
\end{align}
We denote the set of all such coordinates by $ \mathbf{t} \cup \mathbf{h}$, where
\begin{align}
	\mathbf{t}&=\{t_{i,l}|1\le i\le m,\ l\in \mathbb{Z}\}, \label{flatt}\\
	\mathbf{h}&=\{h_{k,r}|1\le k\le m,\ r\in\{0,1\}\}. \label{flath}
\end{align}
Under this coordinate system, the metric $\eta$ takes the following constant forms:
	\begin{align}
	&\left\langle\frac{\partial}{\partial t_{i,l}}, \frac{\partial}{\partial t_{i,l'}}\right\rangle_\eta=-d_{i} \delta_{-d_{i}, l+l'}, \quad i\in\left\{1,\cdots,m\right\},\ l,l'\in\mathbb{Z}; \label{etaflat1}\\
	&\left\langle\frac{\partial}{\partial h_{k,r}}, \frac{\partial}{\partial h_{k,r'}}\right\rangle_\eta=\delta_{1, r+r'},\quad k\in\left\{1,\cdots,m\right\},\ r,r'\in\left\{0,1\right\}, \label{etaflat3}
\end{align}
and any other pairing between these basis vectors vanishes, confirming that $\mathbf{t}\cup\mathbf{h}$ constitute a complete system of flat coordinates for the metric $\eta$.

The flat vector fields corresponding to the coordinates $\mathbf{t}$ and $\mathbf{h}$ are given by the following expressions:
		\begin{align}
		&	\frac{\p \vec{a}}{\p t_{i,l}}=\left(-(\zeta(z)^{\frac{l}{d_{i}}}\zeta'(z)\mathbf{1}_{\gamma_{i}})_{-}, (\zeta(z)^{\frac{l}{d_{i}}}\zeta'(z)\mathbf{1}_{\gamma_{i}})_{+}\right),\quad 1\le i\le m, ~ l\in\mathbb{Z}; \label{flatcom1} \\
		&	\frac{\p \vec{a}}{\p h_{k,r}}=\left(-(\ell'(z)\ell(z)^{-r}\mathbf{1}_{\gamma_{k}})_{-},-(\ell'(z)\ell(z)^{-r}\mathbf{1}_{\gamma_{k}})_{-}\right),\quad 1\le k\le m, ~ 0\le r\le 1,\label{flatcom3}
	\end{align}
	where $\mathbf{1}_{\gamma_k}\in\mathcal{H}$ denotes the characteristic function of the boundary component $\gamma_k$. 
	Specifically, we have
	\begin{align*}
			&\frac{\p a(z)}{\p h_{k,0}}=\frac{\p \hat{a}(z)}{\p h_{k,0}}=-(\hat{a}'(z))_{\varphi_{k},\le -1},\quad 	\frac{\p a(z)}{\p h_{k,1}}=\frac{\p \hat{a}(z)}{\p h_{k,1}}=\frac{1}{z-\varphi_{k}},
	\end{align*}
	where $(\cdot)_{\varphi_k, \le -1}$ denotes the principal part of the Laurent expansion at the pole $\varphi_k$. 
	
	The unit vector field $e$ is given by
	$$e=\sum_{j=1}^{m}\left(\frac{\partial}{\partial t_{j,0}} + \frac{\partial}{\partial h_{j,0}}\right)$$ and satisfies $$\partial_e a(z) = 1 - a'(z), \quad \partial_e \hat{a}(z) = 1 - \hat{a}'(z).$$
	\subsection{Principal hierarchy and $\tau$-structure}
The Hamiltonian densities $\{ \theta_{\alpha, p} \}$ for the principal hierarchy associated with the flat coordinates $\mathbf{t}$ and $\mathbf{h}$ are defined by the following residue and contour integral expressions:
\begin{align*}
	\theta_{h_{k,0},p}=&\frac{1}{(1+p)!} \res_{z=\varphi_{k}} \hat{a}(z)^{1+p} d z,\\
	\theta_{h_{k,1},p} =& \frac{1}{2\pi\mathbf{i}}\oint_{\gamma_{k}} \frac{\hat{a}(z)^{p} - a(z)^p}{p!} \log \frac{\zeta(z)^{\frac{1}{d_{k}}}}{z-\varphi_{k}} dz + \res_{z=\varphi_{k}} \frac{\hat{a}(z)^{p}}{p!} \left( \log \left( \hat{a}(z)(z-\varphi_{k}) \right) - c_{p} \right) dz\\
	&- \res_{z=\infty} \frac{a(z)^{p}}{p!} \left( \log \frac{a(z)}{z-\varphi_{k}} - c_{p} \right) dz + \frac{1}{2 \pi \mathbf{i}} \sum_{i' \neq k} \oint_{\gamma_{i'}} \frac{a(z)^p}{p!} \log(z-\varphi_{k}) dz,\\
	\theta_{t_{i,l},p} =& \frac{1}{2 \pi \mathbf{i}} \frac{1}{(p+1) !} \frac{d_{i}}{l+d_{i}} \oint_{\gamma_{i}} \zeta(z)^{\frac{l}{d_{i}}} \left( a(z)^{p+1} - \hat{a}(z)^{p+1} \right) d z,\\
	\theta_{t_{i,-d_{i}},p} =& \frac{d_{i}}{2 \pi \mathbf{i}} \oint_{\gamma_{i}} \frac{a(z)^p}{p !} \log \frac{\zeta(z)^{\frac{1}{d_{i}}}}{z-\varphi_{i}} d z + d_{i} \res_{z=\infty} \frac{a(z)^{p}}{p!} \left( \log \frac{a(z)}{z-\varphi_{i}} - c_{p} \right) d z \\
	&- \frac{d_{i}}{2 \pi \mathrm{i}} \sum_{i' \neq i} \oint_{\gamma_{i'}} \frac{a(z)^p}{p!} \log(z-\varphi_{i}) d z
\end{align*}
where $c_p = \sum_{j=1}^p \frac{1}{j}$ and $c_{0}=0$. Accordingly, the Hamiltonian density associated with the unit vector field $e$ is given by:
$$
\theta_{e,p} = \sum_{k=1}^{m} \left( \theta_{t_{k,0},p} + \theta_{h_{k,0},p} \right)=-\frac{1}{(1+p)!} \res_{z=\infty} a(z)^{1+p} d z.
$$

Suppose that the Laurent series $\lambda_i(z)$ are defined by the expansions of $a(z)$ and $\hat{a}(z)$ in the vicinity of their respective poles:
\begin{equation}\label{assum}
	a(z) = \lambda_0(z) \hbox{ as } z\to\infty; \qquad \hat{a}(z)= \lambda_j(z) \hbox{ as } z\to \varphi_{j},\quad j = 1, \ldots, m.
\end{equation}
Then the genus-zero universal Whitham hierarchy is characterized by the following Lax equations:
\begin{align}
	& 
	\frac{\partial \lambda_{i}(z)}{\partial \sigma^{0,k}}=\{\left(\lambda_{0}(z)^k\right)_{\infty,\,\ge 0}, \lambda_{i}(z)\}, \label{whi1}\\
	& 
	\frac{\partial \lambda_{i}(z)}{\partial \sigma^{j,k}}=\{-(\lambda_{j}(z)^k)_{\varphi_{j},\,< 0}, \lambda_{i}(z)\},\label{whi2}
	\\
	& 
	\frac{\partial \lambda_{i}(z)}{\partial \sigma^{j,0}}=\{\log (z-\varphi_{j}), \lambda_{i}(z)\}, \label{whi3}
\end{align}
where $\{ \cdot, \cdot \}$ denotes the Poisson bracket $\{f,g\}=\p_{z}f\p_{x}g-\p_{x}f\p_{z}g$
with $x=\sigma^{0,1}$.
The flows of the principal hierarchy on $\mathcal{M}$ are identified with the Whitham hierarchy flows via
\begin{align*}
	&\frac{1}{(p+1)!}\frac{\p}{\p \sigma^{i,p+1}}=\frac{\p}{\p T^{h_{i,0},p}}, \quad \frac{\p}{\p \sigma^{i,0}}=\frac{\p}{\p T^{h_{i,1},0}},\quad i=1,\dots,m; \\
	&\frac{1}{(p+1)!}\frac{\p}{\p \sigma^{0,p+1}}=\frac{\p}{\p T^{e,p}}:=\dsum_{k=1}^{m}\left(\frac{\p}{\p T^{t_{k,0},p}}+\frac{\p}{\p T^{h_{k,0},p}}\right).
\end{align*}

The following theorem provides the explicit expression of the $\tau$-structure for the principal hierarchy when restricted to the sector corresponding to the universal Whitham hierarchy.
\begin{thm}\label{taufor}
The functions $\Omega_{\alpha, p; \beta, q}$ defined by the following contour integrals constitute the $\tau$-structure for the principal hierarchy over the sector $\{h_{k, 0}\}_{k=1}^m \cup \{e\}:$
	\begin{align*}
		&\Omega_{\alpha,p;\beta,q}= \frac{1}{(1+p)! (1+q)!}\frac{1}{2\pi\mathrm{i}}\oint_{\gamma}(Q_{\alpha,p})_{-}d Q_{\beta,q},\quad \alpha,\beta\in \{h_{k,0}\}_{k=1}^{m},\\
		&\Omega_{e,p;\beta,q}= -\frac{1}{(1+p)! (1+q)!}\frac{1}{2\pi\mathrm{i}}\oint_{\gamma}(Q_{e,p})_{+}d Q_{\beta,q},\quad \beta\in \{h_{k,0}\}_{k=1}^{m} \cup\{e\},
	\end{align*}
	where
	$$Q_{e, p} = a(z)^{p + 1} \sum_{k=1}^m \mathbf{1}_{\gamma_k}, \quad Q_{h_{k, 0}, p} = \hat{a}(z)^{p + 1} \mathbf{1}_{\gamma_k}$$and $\gamma = \bigcup_{k=1}^m \gamma_k$ denotes the union of the boundaries of the disks $D_k$.
\end{thm}

\begin{proof}
	The initial conditions \eqref{tausym} are satisfied by direct verification of the explicit formula for $\Omega_{\alpha, p; \beta, q}$  and $\theta_{\alpha, p}$. To verify the recursion relation \eqref{taucru}:
	$$\Omega_{\alpha, p-1; \beta, q} + \Omega_{\alpha, p; \beta, q-1} = \langle d\theta_{\alpha, p}, \eta(d\theta_{\beta, q}) \rangle,$$
	we consider the differentials $d\theta_{\alpha, p} \in \mathcal{H} \times \mathcal{H}$ defined as
	$$d\theta_{\alpha, p} = \frac{1}{(p + 1)!} \left( \frac{\partial Q_{\alpha, p}}{\partial a}, \frac{\partial Q_{\alpha, p}}{\partial \hat{a}} \right)$$
	for  $\alpha,\beta\in \{h_{k,0}\}_{k=1}^{m}$.
	
	 Utilizing the derivative relations:
	   	$$
	\frac{\p Q_{\alpha,p}}{\p a}+\frac{\p Q_{\alpha,p}}{\p \hat{a}}=(p+1)Q_{\alpha,p-1},\quad Q_{\alpha,p}'=\frac{\p Q_{\alpha,p}}{\p a} a'(z) + \frac{\p Q_{\alpha,p}}{\p \hat{a}} \hat{a}'(z),
	$$
	the right-hand side (RHS) is expanded using the explicit form of the map $\eta$ defined in \eqref{etaom} and the pairing $\langle \cdot, \cdot \rangle$ defined in \eqref{pairmec}:
		$$\begin{aligned}
		\langle d\theta_{\alpha, p}, \eta(d\theta_{\beta, q}) &= C_{p,q} \oint_{\gamma} \Bigg\{ \frac{\p Q_{\alpha,p}}{\p a} \left[ a'(z) \left(\frac{\p Q_{\beta,q}}{\p a} + \frac{\p Q_{\beta,q}}{\p \hat{a}}\right)_{-} - \left(\frac{\p Q_{\beta,q}}{\p a} a'(z) + \frac{\p Q_{\beta,q}}{\p \hat{a}} \hat{a}'(z)\right)_{-} \right] \\
		&\quad + \frac{\p Q_{\alpha,p}}{\p \hat{a}} \left[ -\hat{a}'(z) \left(\frac{\p Q_{\beta,q}}{\p a} + \frac{\p Q_{\beta,q}}{\p \hat{a}}\right)_{+} + \left(\frac{\p Q_{\beta,q}}{\p a} a'(z) + \frac{\p Q_{\beta,q}}{\p \hat{a}} \hat{a}'(z)\right)_{+} \right] \Bigg\} dz\\
		&= C_{p,q-1} \oint_{\gamma} \left[ \frac{\partial Q_{\alpha,p}}{\partial a} a' (Q_{\beta,q-1})_{-} - \frac{\partial Q_{\alpha,p}}{\partial \hat{a}} \hat{a}' (Q_{\beta,q-1})_{+} \right] dz\\
		&\quad +C_{p,q} \oint_{\gamma} \left[ -\frac{\partial Q_{\alpha,p}}{\partial a} (Q_{\beta,q}')_{-} + \frac{\partial Q_{\alpha,p}}{\partial \hat{a}} (Q_{\beta,q}')_{+} \right] dz\\
		&= C_{p,q-1} \oint_{\gamma} \left[ \frac{\partial Q_{\alpha,p}}{\partial a} a' (Q_{\beta,q-1})_{-} + \frac{\partial Q_{\alpha,p}}{\partial \hat{a}} \hat{a}' (Q_{\beta,q-1})_{-} \right] dz\\
		&\quad +C_{p,q} \oint_{\gamma} \left[ -\frac{\partial Q_{\alpha,p}}{\partial a} (Q_{\beta,q}')_{-} - \frac{\partial Q_{\alpha,p}}{\partial \hat{a}} (Q_{\beta,q}')_{-} \right] dz\\
		&=C_{p,q-1} \oint_{\gamma}Q_{\alpha,p}'(Q_{\beta,q-1})_{-}dz+C_{p-1,q}\oint_{\gamma}Q_{\alpha,p-1}'(Q_{\beta,q})_{-}dz,
	\end{aligned}$$
	where we denote
	$$C_{p, q} = \frac{1}{(1 + p)! (1 + q)!}.$$
	By applying the contour identity $\oint_{\gamma} f_+ dz = - \oint_{\gamma} f_- dz$ which holds when $f$ is a total derivative, we can rearrange the terms to obtain:
	$$\langle d\theta_{\alpha, p}, \eta(d\theta_{\beta, q}) = C_{p, q-1} \oint_{\gamma} Q_{\alpha, p}' (Q_{\beta, q-1})_- dz + C_{p-1, q} \oint_{\gamma} Q_{\alpha, p-1}' (Q_{\beta, q})_- dz.$$
	This expression coincides with the definition of $\Omega_{\alpha, p; \beta, q-1} + \Omega_{\alpha, p-1; \beta, q}$ in the explicit formula, thereby completing the proof for this sector. 
	
	The remaining cases can be established through a similar procedure.
\end{proof}
The components of the $\tau$-structure involving the logarithmic sector $h_{k, 1}$ are determined by the derivatives of the Hamiltonian densities:
$$\Omega_{\alpha, p; h_{k, 1}, 0} = \frac{\partial \theta_{\alpha, p + 1}}{\partial h_{k, 1}}.$$
Evaluating this derivative for the sector associated with $e$ yields:
$$\Omega_{e, p; h_{k, 1}, 0} = - \frac{1}{(p + 1)!} \mathrm{Res}_{z = \infty} \frac{a(z)^{p + 1}}{z - \varphi_k} dz.$$
Similarly, for the sectors associated with $h_{j, 0}$, we obtain: $$\Omega_{h_{j, 0}, p; h_{k, 1}, 0} = \frac{1}{(p + 1)!} \mathrm{Res}_{z = \varphi_j} \frac{\hat{a}(z)^{p + 1}}{z - \varphi_k} dz.$$
In particular,
$$\Omega_{h_{j, 1}, 0; h_{k, 1}, 0} = \log \left( \varphi_j - \varphi_k \right), \quad j \neq k,$$$$\Omega_{h_{k, 1}, 0; h_{k, 1}, 0} = \log \hat{a}_{k, -1},$$
where $\hat{a}_{k, -1}$ denotes the residue coefficient defined in the Laurent expansion of $\hat{a}(z)$ near $\varphi_k$.

\subsection{Properties of the $\tau$-structure}\label{whmtau}
Consider the expansions of $z$ in terms of the functions $a$ and $\hat{a}$ near their respective poles:
$$z = a - u_{0,1} a^{-1} - u_{0,2} a^{-2} - \cdots, \quad z \to \infty,$$
$$z = u_{j,0} + u_{j,1} \hat{a}^{-1} + u_{j,2} \hat{a}^{-2} - \cdots, \quad z \to \varphi_{j}.$$
The expansion coefficients $\{u_{k, n}\}$ form a coordinate system on $\mathcal{M}$. They are directly related to the Hamiltonian densities via the following identities:
$$\theta_{e, p} = \frac{1}{p!} u_{0, p+1}, \quad \theta_{h_{k, 0}, p} = \frac{1}{p!} u_{k, p+1}, \quad \theta_{h_{k, 1}, 0} = u_{k, 0},\quad k = 1, \ldots, m.$$
The partial derivatives of $\vec{a}$ with respect to these coordinates are expressed as:
$$\frac{\partial \vec{a}}{\partial u_{0, j}} = ( a(z)^{-j} a'(z), 0 ),\quad \frac{\partial \vec{a}}{\partial u_{k, j}} = ( 0, - ( \hat{a}(z)^{-j} \hat{a}'(z) ) \mathbf{1}_{\gamma_{k}} ).$$
In terms of the coordinates $\{u_{k, n}\}$, the components of the $\tau$-structure in the logarithmic sector take the following explicit form:
$$\Omega_{h_{k,1}, 0; h_{k',1}, 0} = \log(u_{k,0} - u_{k',0}),\quad k\ne k',\quad \Omega_{h_{k,1}, 0; h_{k,1}, 0} = \log u_{k,1}.$$

The following property of the $\tau$-structure is essential for deriving the stability conditions for the derivatives of the prepotential associated with $H_{0; \mathbf{n}}$ in the subsequent subsection.
\begin{lem}\label{taufini}
The components of the $\tau$-structure are determined by a finite number of expansion coefficients $\{ u_{k, j} \}$. Specifically, the dependence relations are as follows:

	\begin{tabular}{l l}
	$ \Omega_{h_{k,0}, p; h_{k,0}, q} $ & depends on $ \{ u_{k, j} \mid 0 \le j \le p+q+3 \}; $ \\
	$ \Omega_{e, p; e, q} $             & depends on $ \{ u_{0, j} \mid 1 \le j \le p+q+1 \}; $ \\
	$ \Omega_{e, p; h_{k,0}, q} $       & depends on $ \{ u_{0, j} \mid 1 \le j \le p \} \cup \{ u_{k, j} \mid 0 \le j \le q+1 \}; $ \\
	$ \Omega_{h_{k_{1}, 0}, p; h_{k_{2}, 0}, q} $ & depends on $ \{ u_{k_{1}, j} \mid 0 \le j \le p+1 \} \cup \{ u_{k_{2}, j} \mid 0 \le j \le q+1 \} $, for $ k_{1} \neq k_{2}; $ \\
	$ \Omega_{h_{k_{1}, 1}, 0; h_{k_{2}, 0}, q} $ & depends on $ \{ u_{k_{1}, 0} \} \cup \{ u_{k_{2}, j} \mid 0 \le j \le q+1 \} $, for $ k_{1} \neq k_{2}; $ \\
	$ \Omega_{h_{k, 1}, 0; h_{k, 0}, q} $ & depends on $ \{ u_{k, j} \mid 0 \le j \le q+2 \}; $\\
		$ \Omega_{h_{k, 1}, 0; e, q} $ & depends on $ \{ u_{0, j} \mid 1 \le j \le q \}. $
\end{tabular}
\end{lem}
\begin{proof}
	Using the explicit formulas in Theorem \ref{taufor}, we have:
	\begin{align*}
		\Omega_{h_{k,0}, p; h_{k,0}, q} = &C_{p,q}\oint_{\gamma} \left( \hat{a}(z)^{1+p} \textbf{1}_{\gamma_{k}}\right)_{-} d \left( \hat{a}(z)^{1+q} \textbf{1}_{\gamma_{k}}\right)\\
		=&C_{p,q-1}\res_{\varphi_{k}} \left( \hat{a}(z)^{1+p} \right)_{\varphi_{k},\le -1} \hat{a}(z)^{q}\hat{a}'(z)d z.
	\end{align*}
This expression is determined solely by the coefficients $\{ u_{k, j} \mid j \ge 0 \}$. To establish the upper bound on this dependence, we consider the derivative with respect to $u_{k, j}$:
\begin{align*}
	\frac{\p 	\Omega_{h_{k,0}, p; h_{k,0}, q}}{\p u_{k,j}}=&C_{p-1,q}\oint_{\gamma} \left( \hat{a}(z)^{p-j}\hat{a}'(z) \textbf{1}_{\gamma_{k}}\right)_{-} d \left( \hat{a}(z)^{1+q} \textbf{1}_{\gamma_{k}}\right)\\
	&+C_{p,q-1}\oint_{\gamma} \left( \hat{a}(z)^{1+p} \textbf{1}_{\gamma_{k}}\right)_{-} d \left( \hat{a}(z)^{q-j}\hat{a}'(z) \textbf{1}_{\gamma_{k}}\right).
\end{align*}
For $j > p + q + 3$, these integrals vanish as the order of the pole at $\varphi_k$ becomes non-positive.

Similarly, for the cross-sector terms where $k_1 \neq k_2$, the integral reduces to:
	\begin{align*}
		\Omega_{h_{k_{1},0}, p; h_{k_{2},0}, q} = &C_{p,q}\oint_{\gamma} \left( \hat{a}(z)^{1+p} \textbf{1}_{\gamma_{k_{1}}}\right)_{-} d \left( \hat{a}(z)^{1+q} \textbf{1}_{\gamma_{k_{2}}}\right)\\
		=&C_{p,q-1}\res_{\varphi_{k_{2}}} \left( \hat{a}(z)^{1+p} \right)_{\varphi_{k_{1}},\le -1} \hat{a}(z)^{q}  \hat{a}'(z)dz \\
		=&C_{q,p-1}\res_{\varphi_{k_{1}}} \left( \hat{a}(z)^{1+q} \right)_{\varphi_{k_{2}},\le -1} \hat{a}(z)^{p}  \hat{a}'(z)dz.
	\end{align*}
The evaluation of this residue depends solely on the principal part of the expansion at $\varphi_{k_1}$ and expansion at $\varphi_{k_2}$. By symmetry, this constrains the dependency to the indices $\{ u_{k_1, j} \mid 0 \le j \le p + 1 \}$ and $\{ u_{k_2, j} \mid 0 \le j \le q + 1 \}$, respectively.

The remaining cases follow from analogous residue calculations.
\end{proof}
In view of the finite dependency properties established in Lemma \ref{taufini}, we define $\mathcal{M}^{\mathrm{formal}}$ as the space consisting of $m + 1$ formal Laurent series of the following form:
\begin{align*}
	&\lambda_{0}(z)=z+a_{1}z^{-1}+a_{2}z^{-2}+\cdots,\\	
	&\lambda_{j}(z)=\hat{a}_{j,-1}(z-\varphi_{j})^{-1}+\hat{a}_{j,0} +\hat{a}_{j,1} (z-\varphi_{j})+\cdots, \quad j=1,\cdots,m.
\end{align*}
The principal hierarchy $\partial / \partial T^{\alpha, p}$ and its associated $ \tau $-structure are thus well-defined on $\mathcal{M}^{\mathrm{formal}}$ for the sector $(\alpha, p) \in \{ (h_{k, 0}, p), (h_{k, 1}, 0), (e, p) \}$. The $\tau$-cover for this hierarchy, which coincides with the genus-zero universal Whitham hierarchy, are given by:
$$\frac{\partial^2 \mathcal{F}}{\partial T^{\alpha, p} \partial T^{\beta, q}} = \Omega_{\alpha, p; \beta, q}(\mathbf{u}).$$
In this system, the coordinates $\mathbf{u} = \{u_{j, k}\}$ on the right-hand side are identified with the following derivatives of the prepotential $\mathcal{F}$:
\begin{equation}\label{replu}
\left\{ \begin{array}{ll} 
	u_{k, 0} \longmapsto \frac{\partial^2 \mathcal{F}}{\partial x \partial T^{h_{k, 1}, 0}} & (k \ge 1) \\[10pt]
	u_{k, j} \longmapsto (j - 1)! \frac{\partial^2 \mathcal{F}}{\partial x \partial T^{h_{k, 0}, j - 1}} & (k \ge 1, j \ge 1) \\[10pt]
	u_{0, j} \longmapsto (j - 1)! \frac{\partial^2 \mathcal{F}}{\partial x \partial T^{e, j - 1}} & (j \ge 1)
\end{array} \right. 
\end{equation}
where the spatial derivative $\partial / \partial x$ is identified with the unit flow:$$\frac{\partial}{\partial x} = \frac{\partial}{\partial T^{e, 0}}.$$

In \cite{takasaki2007universal}, this system is formulated in Hirota bilinear form and is shown to be equivalent to the dispersionless limit of the multicomponent KP hierarchy (see also \cite{teo2011multicomponent}).
\section{Stabilization of $H_{0; \mathbf{n}}$ and Related Hierarchies}

In this section, we revisit the Hurwitz Frobenius manifold $H_{0; \mathbf{n}}$ and the integrable hierarchy constructed by Basalaev, which is characterized by the intrinsic stability properties of the prepotential's derivatives. By leveraging the $\tau$-structure of the infinite-dimensional manifold $\mathcal{M}$ established in the preceding sections, we provide a direct identification between Basalaev's hierarchy and the genus-zero universal Whitham hierarchy. Notably, our approach offers an intrinsic derivation that is independent of the explicit formulations of the $\tau$-cover in \cite{takasaki2007universal}.

Applying the same methodology, we further establish an equivalence between the hierarchy arising from the stability of open WDVV solutions for $H_{0;\mathbf{n}}$ and the multi-component dispersionless modified KP (mKP) hierarchy. This result constitutes a multicomponent generalization of the integrable structures introduced in \cite{basalaev2022integrable}.

\subsection{Explicit formulas for the prepotential derivatives of $H_{0;\mathbf{n}}$}

Recall that for a given positive integer $m$ and a sequence of positive integers $\mathbf{n} = \left( n_0, n_1, \dots, n_m \right)$, the Hurwitz space $H_{0; \mathbf{n}}$ is defined as the set of functions of the form:
\[
\lambda(z) = z^{n_0} + a_{0,n_0-2}z^{n_0-2} + \cdots + a_{0,0} + \sum_{i=1}^{m}\sum_{j=1}^{n_i}a_{i,j}(z-a_{i,0})^{-j}
\]
where $a_{i,n_i} \neq 0$ for $i = 1, \dots, m$. The set of parameters $\{a_{0,i}\}_{i=0}^{n_0-2} \cup \{a_{i,j} \mid 1 \leq i \leq m, 0 \leq j \leq n_i\}$ constitutes a coordinate system on $H_{0;\mathbf{n}}$.

The space $H_{0; \mathbf{n}}$ carries a canonical Frobenius structure defined as follows.
For any \( \p', \p'', \p''' \in T_{\lambda(z)}H_{0;\mathbf{n}} \), the metric $\eta$ is given by the residue formula:
\[
\langle \p', \p'' \rangle_{\eta} := \eta(\p', \p'') = \sum_{|\lambda|<\infty} \res_{d\lambda=0} \frac{\p'(\lambda(z)dz) \p''(\lambda(z)dz)}{d\lambda(z)}
\]
and the $(0,3)$-type tensor $c$ is defined by:
\[
c(\p', \p'', \p''') := \sum_{|\lambda|<\infty} \res_{d\lambda=0} \frac{\p'(\lambda(z)dz) \p''(\lambda(z)dz) \p'''(\lambda(z)dz)}{d\lambda(z)dz}.
\]
The multiplication $\circ$ on the tangent space is uniquely determined by the relation $\eta(\partial' \circ \partial'', \partial''') = c(\partial', \partial'', \partial''')$. Furthermore, the unit vector field $e$ and the Euler vector field $E$ on $H_{0; \mathbf{n}}$ are characterized by:
$$\mathcal{L}_{e} \lambda(z) = 1, \quad \mathcal{L}_{E} \lambda(z) = \lambda(z) - \frac{z}{n_0} \lambda'(z)$$
 where $\mathcal{L}$ denotes the Lie derivative. This data endows $H_{0; \mathbf{n}}$ with the structure of a semisimple Frobenius manifold of conformal dimension $d = 1 - \frac{2}{n_0}$.

The flat coordinates of the metric $ \eta $, denoted by
\[
\mathbf{v} = \{v_{0,j}\}_{j=1}^{n_{0}-1} \cup \{v_{1,j}\}_{j=0}^{n_{1}} \cup \cdots \cup \{v_{m,j}\}_{j=0}^{n_{m}},
\]
are defined via the coefficients of the local expansions of the function $z(\lambda)$ near the poles of $\lambda(z)$. Specifically, we have:
\[
z = \left\{
\begin{aligned}
	&v_{i,0} + v_{i,1}\lambda^{-\frac{1}{n_{i}}} + \cdots, \quad &z &\to a_{i,0},\ i = 1, \cdots, m, \\
	&\lambda^{\frac{1}{n_{0}}} - v_{0,1}\lambda^{-\frac{1}{n_{0}}} - v_{0,2}\lambda^{-\frac{2}{n_{0}}} + \cdots, \quad &z &\to \infty,\ \ i = 0.
\end{aligned}
\right.
\]
The partial derivatives of $\lambda(z)$ with respect to these flat coordinates are given by:
\[
\frac{\partial \lambda(z)}{\partial v_{i,j}} = \left\{
\begin{aligned}
	&-(\lambda(z)^{-\frac{j}{n_{i}}}\lambda'(z))_{a_{i,0},\le -1},\quad i = 1, \cdots, m,\ j = 0, \cdots, n_{i}, \\
	&(\lambda(z)^{-\frac{j}{n_{0}}}\lambda'(z))_{\infty,\ge 0},\quad i = 0,\ j = 1, \cdots, n_{0}-1.
\end{aligned}
\right.
\]

Following \cite{aoyama1996topological,ma2023principal}, the Hamiltonian densities $\theta_{\alpha, p}$ for the Frobenius manifold $H_{0; \mathbf{n}}$ are defined via the following residues of $\lambda(z):$
\begin{align*}
	\theta_{v_{0,j},p}(z) =&-\frac{\Gamma\left(1-\frac{j}{n_{0}}\right)}{\Gamma\left(2+p-\frac{j}{n_{0}}\right)} \res_{z=\infty} \lambda(z)^{1+p-\frac{j}{n_{0}}} d z, \quad j = 1, \ldots, n_{0} - 1, \\
	\theta_{v_{k,j},p}(z) =&\frac{\Gamma\left(1-\frac{j}{n_{k}}\right)}{\Gamma\left(2+p-\frac{j}{n_{k}}\right)} \res_{z=a_{k,0}} \lambda(z)^{1+p-\frac{j}{n_{k}}} d z, \quad k = 1, \ldots, m, \ j = 0, \ldots, n_{k} - 1,\\
	\theta_{v_{k,n_{k}},p}(z) =& n_{k}\res_{z=a_{k,0}} \frac{\lambda(z)^{p}}{p!} \left( \log \left( \lambda(z)^{\frac{1}{n_{k}}}(z-a_{k,0}) \right) - \frac{c_{p}}{n_{k}} \right) dz\\
	&- n_{k}\res_{z=\infty} \frac{\lambda(z)^{p}}{p!} \left( \log \frac{\lambda(z)^{\frac{1}{n_{0}}}}{z-a_{k,0}} - \frac{c_{p}}{n_{0}} \right) dz + n_{k} \sum_{i'\ne k}\res_{z=a_{i',0}} \frac{\lambda(z)^p}{p!} \log(z-a_{k,0}) dz.
\end{align*}

By invoking the identities in \eqref{tausym}, the primary components of the $\tau$-structure associated with the flat coordinates $\{v_{k, j}\}$ can be evaluated explicitly:
$$\Omega_{v_{0, j}, 0; v_{0, l}, 0} = \frac{n_0^2}{(n_0 - j)(n_0 - l)} \mathrm{Res}_{z = \infty} \left( \lambda(z)^{1 - \frac{j}{n_0}} \right)_{\infty, \ge 0} d \left( \lambda(z)^{1 - \frac{l}{n_0}} \right),$$
$$\Omega_{v_{k, j}, 0; v_{k', l}, 0} = \frac{n_k n_{k'}}{(n_k - j)(n_{k'} - l)} \mathrm{Res}_{z = a_{k', 0}} \left( \lambda(z)^{1 - \frac{j}{n_k}} \right)_{a_{k, 0}, \le -1} d \left( \lambda(z)^{1 - \frac{l}{n_{k'}}} \right),$$
$$\Omega_{v_{0, j}, 0; v_{k, l}, 0} = - \frac{n_0 n_k}{(n_0 - j)(n_k - l)} \mathrm{Res}_{z = a_{k, 0}} \left( \lambda(z)^{1 - \frac{j}{n_0}} \right)_{\infty, \ge 0} d \left( \lambda(z)^{1 - \frac{l}{n_k}} \right).$$
$$\Omega_{v_{0, j}, 0; v_{k, n_k}, 0} = - \frac{n_0 n_k}{n_0 - j} \mathrm{Res}_{z = \infty} \lambda(z)^{1 - \frac{j}{n_0}} \frac{1}{z - a_{k, 0}} dz,$$
$$\Omega_{v_{k, j}, 0; v_{k', n_{k'}}, 0} = \frac{n_k n_{k'}}{n_k - j} \mathrm{Res}_{z = a_{k, 0}} \lambda(z)^{1 - \frac{j}{n_k}} \frac{1}{z - a_{k', 0}} dz.$$
$$\Omega_{v_{k, n_k}, 0; v_{k', n_{k'}}, 0} = n_k n_{k'} \log \left( a_{k, 0} - a_{k', 0} \right), \quad k \neq k',$$
$$\Omega_{v_{k, n_k}, 0; v_{k, n_k}, 0} = n_k^2 \log v_{k, 1}.$$
These components provide the second-order derivatives of the prepotential $F_{H_{0; \mathbf{n}}}$ via the following identification:
$$
\frac{\partial^{2} F_{H_{0;\mathbf{n}}}}{\partial v_{i,j} \partial v_{k,l}}=\Omega_{v_{i,j},0;v_{k,l},0}.
$$

\subsection{Stabilization of the Prepotential}\label{embd}
Consider the map from the Hurwitz manifold $H_{0;\mathbf{n}}$ to the formal space $\mathcal{M}^{\mathbf{formal}}$ defined by $\lambda_{i}(z) = \lambda(z)^{i/n_{i}}$ as $z \to a_{i, 0}$. This map induces the coordinate identification $u_{i, j} = u_{i, j}(\mathbf{v})$, which specifically satisfies:\\
$$u_{0,j} = v_{0,j}, \quad 1 \le j \le n_0-1,$$
$$u_{k,j} = v_{k,j}, \quad 1 \le k \le m, \quad 0 \le j \le n_k.$$
Under these identifications, the functions $\Omega_{v_{i, j}, 0; v_{k, r}, 0}$ on $H_{0;\mathbf{n}}$ are identified with the formal $\tau$-structure on $\mathcal{M}^{\mathrm{formal}}$ as follows:
$$\frac{1}{n_0^2} \Omega_{v_{0, n_0 - p}, 0; v_{0, n_0 - q}, 0}(\mathbf{v}) = (p - 1)! (q - 1)! \Omega_{e, p - 1; e, q - 1}(\mathbf{u}) \left. \right|_{u = v}, \quad n_0 \ge p + q ;$$
$$\frac{1}{n_k^2} \Omega_{v_{k, n_k - p}, 0; v_{k, n_k - q}, 0}(\mathbf{v}) = (p - 1)! (q - 1)! \Omega_{h_{k, 0}, p - 1; h_{k, 0}, q - 1}(\mathbf{u}) \left. \right|_{u = v}, \quad n_k \ge p + q + 1;$$
$$\frac{1}{n_k n_{k'}} \Omega_{v_{k, n_k - p}, 0; v_{k', n_{k'} - q}, 0}(\mathbf{v}) = (p - 1)! (q - 1)! \Omega_{h_{k, 0}, p - 1; h_{k', 0}, q - 1}(\mathbf{u}) \left. \right|_{u = v}, \quad n_k \ge p, \ n_{k'} \ge q,\ k\ne k';$$
$$\frac{1}{n_0 n_k} \Omega_{v_{0, n_0 - p}, 0; v_{k, n_k - q}, 0}(\mathbf{v}) = (p - 1)! (q - 1)! \Omega_{e, p - 1; h_{k, 0}, q - 1}(\mathbf{u}) \left. \right|_{u = v}, \quad n_0 \ge p, \ n_k \ge q;$$
$$\frac{1}{n_0 n_k} \Omega_{v_{0, n_0 - p}, 0; v_{k, n_k}, 0}(\mathbf{v}) = (p - 1)! \Omega_{e, p - 1; h_{k, 1}, 0}(\mathbf{u}) \left. \right|_{u = v}, \quad n_0 \ge p+1;$$
$$\frac{1}{n_k n_{k'}} \Omega_{v_{k, n_k - p}, 0; v_{k', n_{k'}}, 0}(\mathbf{v}) = (p - 1)! \Omega_{h_{k, 0}, p - 1; h_{k', 1}, 0}(\mathbf{u}) \left. \right|_{u = v}, \quad n_k \ge p,\ k\ne k';$$
$$\frac{1}{n_k^2} \Omega_{v_{k, n_k}, 0; v_{k, n_k - p}, 0}(\mathbf{v}) = (p - 1)! \Omega_{h_{k, 1}, 0; h_{k, 0}, p - 1}(\mathbf{u}) \left. \right|_{u = v}, \quad n_k \ge p + 1;$$
$$\frac{1}{n_k n_{k'}} \Omega_{v_{k, n_k}, 0; v_{k', n_{k'}}, 0}(\mathbf{v}) = \Omega_{h_{k, 1}, 0; h_{k', 1}, 0}(\mathbf{u}) \left. \right|_{u = v}.$$

The preceding identifications lead to the following stabilization theorem for the prepotential $F_H$ of the Hurwitz manifold $H_{0; \mathbf{n}}$:
\begin{thm}\label{preptau}
Let $\hat{v}_{i, j} = n_i v_{i, n_i - j}$ be the rescaled coordinates. The second derivatives of the prepotential $F_H$ are stable in the limit $n \to \infty$, where $n = \min \{ n_0, \ldots, n_m \}$. Specifically:
$$\lim_{n \to \infty} \frac{\partial^2 F_H}{\partial \hat{v}_{0, p} \partial \hat{v}_{0, q}} = (p - 1)! (q - 1)! \Omega_{e, p - 1; e, q - 1},\quad p,q\ge 1;$$
$$\lim_{n \to \infty} \frac{\partial^2 F_H}{\partial \hat{v}_{k, p} \partial \hat{v}_{k', q}} = (p - 1)! (q - 1)! \Omega_{h_{k, 0}, p - 1; h_{k', 0}, q - 1},\quad p,q\ge 1;$$
$$\lim_{n \to \infty} \frac{\partial^2 F_H}{\partial \hat{v}_{0, p} \partial \hat{v}_{k, q}} = (p - 1)! (q - 1)! \Omega_{e, p - 1; h_{k, 0}, q - 1},\quad p,q\ge 1;$$
$$\lim_{n \to \infty} \frac{\partial^2 F_H}{\partial \hat{v}_{0, p} \partial \hat{v}_{k, 0}} = (p - 1)! \Omega_{e, p - 1; h_{k, 1}, 0},\quad p\ge 1;$$
$$\lim_{n \to \infty} \frac{\partial^2 F_H}{\partial \hat{v}_{k, p} \partial \hat{v}_{k', 0}} = (p - 1)! \Omega_{h_{k, 0}, p - 1; h_{k', 1}, 0},\quad p\ge 1;$$
$$\lim_{n \to \infty} \frac{\partial^2 F_H}{\partial \hat{v}_{k, 0} \partial \hat{v}_{k', 0}} = \Omega_{h_{k, 1}, 0; h_{k', 1}, 0}$$
where the right-hand sides are evaluated at $u_{i, j} = v_{i, j} = \frac{1}{n_i} \hat{v}_{i, n_i - j}$.
\end{thm}
From the preceding stabilization theorem and the explicit form of the $\tau$-cover for the Whitham hierarchy discussed in Section \ref{whmtau}, we obtain the following:
\begin{cor}\label{stabhier}
 The $\tau$-cover of the Whitham hierarchy can be equivalently represented by the following system of equations:
	
$$\frac{\partial^2 \mathcal{F}}{\partial \sigma_{i, p} \partial \sigma_{j, q}} = p q \left. \left( \lim_{n \to \infty} \frac{\partial^2 F_{H}}{\partial \hat{v}_{i, p} \partial \hat{v}_{j, q}} \right) \right.,\quad p,q\ge 1;$$
$$\frac{\partial^2 \mathcal{F}}{\partial \sigma_{i, p} \partial \sigma_{j, 0}} = p \left. \left( \lim_{n \to \infty} \frac{\partial^2 F_{H}}{\partial \hat{v}_{i, p} \partial \hat{v}_{j, 0}} \right) \right.,\quad p\ge 1;$$
$$\frac{\partial^2 \mathcal{F}}{\partial \sigma_{i, 0} \partial \sigma_{j, 0}} = \left. \left( \lim_{n \to \infty} \frac{\partial^2 F_{H}}{\partial \hat{v}_{i, 0} \partial \hat{v}_{j, 0}} \right) \right.,$$
where the limits on the right-hand sides are subject to the identifications:
$$v_{k, r} = \frac{1}{r} \frac{\partial^2 \mathcal{F}}{\partial \sigma_{0, 1} \partial \sigma_{k, r}}\ (r\ge 1),\quad v_{k, 0} = \frac{\partial^2 \mathcal{F}}{\partial \sigma_{0, 1} \partial \sigma_{k, 0}}.$$
\end{cor}
Consequently, we recover the results of \cite{basalaev2025genus}.

\subsection{Open extension and stable limits}
In \cite{ma2025solutions}, the author constructed a pair of solutions to the open WDVV equations for $\mathcal{M}$, along with the principal hierarchy for the associated flat $F$-manifold.

For the solution denoted by $F_{\mathcal{M}}^{o}$, the corresponding functions $\tilde{\theta}^{\mathcal{M}}_{\alpha, p}$ over the sectors $\{e\} \cup \mathbf{h} \cup \{s\}$ are given by:
$$\begin{aligned}
	\tilde{\theta}_{e, p}^{\mathcal{M}}
	&=-\frac{1}{(p+1)!} (a(s)^{p+1})_{\infty,\le -1}, \quad \tilde{\theta}_{h_{k,0}, p}^{\mathcal{M}}=-\frac{1}{(p+1)!}(\hat{a}(s)^{p+1})_{\varphi_{k},\le -1},\\
	\tilde{\theta}^{\mathcal{M}}_{s, p}& = \frac{a(s)^{p+1}}{(p+1)!},\quad \tilde{\theta}^{\mathcal{M}}_{h_{k,1}, 0} = \log(s-\varphi_k).
\end{aligned}$$
The associated principal hierarchy is governed by the following system:
$$\frac{\partial \vec{a}}{\partial \widetilde{T}^{h_{k,0},p}}=\frac{\partial \vec{a}}{\partial T^{h_{k,0},p}},\quad \frac{\partial s}{\partial \widetilde{T}^{h_{k,0},p}}=- \frac{1}{(p+1)!} \partial_{x}\Bigl( \hat{a}(s)^{1+p} \Bigr)_{\varphi_{k},\le -1};$$
$$\frac{\partial \vec{a}}{\partial \widetilde{T}^{e,p}}=\frac{\partial \vec{a}}{\partial T^{e,p}},\quad \frac{\partial s}{\partial \widetilde{T}^{e,p}}=- \frac{1}{(p+1)!} \partial_{x}\Bigl( a(s)^{1+p} \Bigr)_{\infty,\le -1};$$
$$\frac{\partial \vec{a}}{\partial \widetilde{T}^{s,p}}=0,\quad \frac{\partial s}{\partial \widetilde{T}^{s,p}}= \frac{1}{(p+1)!} \partial_{x}\Bigl( a(s)^{1+p} \Bigr);$$
$$\frac{\partial \vec{a}}{\partial \widetilde{T}^{h_{k,1},0}}=\frac{\partial \vec{a}}{\partial T^{h_{k,1},0}},\quad \frac{\partial s}{\partial \widetilde{T}^{h_{k,1},0}}= \partial_{x}\Bigl( \log(s-\varphi_k) \Bigr).$$

The $\tau$-structure of this hierarchy is defined by the following system for $( \mathcal{F}, \mathcal{F}^o )$: 
$$\frac{\partial^2 \mathcal{F}}{\partial T^{\alpha, p} \partial T^{\beta, q}} = \Omega_{\alpha, p; \beta, q}(\mathbf{u}),\quad \frac{\p \mathcal{F}}{\p T^{s,0}}=0,$$
$$
\frac{\p \mathcal{F}^{o}}{\p T^{\alpha,p}}=\tilde{\theta}^{\mathcal{M}}_{\alpha,p}(\textbf{u},s),\quad \frac{\p \mathcal{F}^{o}}{\p T^{s,p}}=\tilde{\theta}^{\mathcal{M}}_{s,p}(\textbf{u},s)
$$
where the variables $\mathbf{u}$ are identified through the mapping in \eqref{replu}, and $s = \partial \mathcal{F}^o / \partial x$. 

We now consider the stability of the solution $F_H^o$ to the open WDVV equations for the Hurwitz Frobenius manifold $H_{0; \mathbf{n}}$. Following \cite{ma2025solutions}, the functions $\tilde{\theta}_{v_{i, j}, 0}^H$ associated with the coordinates $\mathbf{v}$ and the open variable $s$ are given by:
$$\tilde{\theta}_{v_{0, j}, 0}^{H} = \frac{n_0}{n_0 - j} \left( \lambda(s)^{\frac{n_0 - j}{n_0}} \right)_{\infty, \ge 0},$$
$$\tilde{\theta}^{H}_{v_{i, j}, 0} = - \frac{n_i}{n_i - j} \left( \lambda(s)^{\frac{n_i - j}{n_i}} \right)_{a_{i, 0}, \le -1},$$
$$\tilde{\theta}^{H}_{v_{i, n_i}, 0} = n_i \log \left( s - a_{i, 0} \right),$$
$$\tilde{\theta}^{H}_{s, 0} = \lambda(s).$$

By utilizing the embedding of $H_{0; \mathbf{n}}$ into $\mathcal{M}^{\mathrm{formal}}$ described in Section \ref{embd}, we establish the following identities for $p \ge 1$:
$$\frac{1}{n_{0}} \tilde{\theta}_{v_{0, n_{0}-p}, 0}^{H}(\mathbf{v},s) = (p-1)! \left( \tilde{\theta}_{s, p-1}^{\mathcal{M}}(\mathbf{u},s) + \tilde{\theta}_{e, p-1}^{\mathcal{M}}(\mathbf{u},s) \right) \Bigg|_{\mathbf{u}=\mathbf{v}},\quad n_0 \ge  p;$$
$$\frac{1}{n_{k}} \tilde{\theta}_{v_{k, n_{k}-p}, 0}^{H}(\mathbf{v},s) = (p-1)! \left( \tilde{\theta}_{h_{k,0}, p-1}^{\mathcal{M}}(\mathbf{u},s)  \right) \Bigg|_{\mathbf{u}=\mathbf{v}},\quad n_k \ge p;$$
$$\frac{1}{n_{i}} \tilde{\theta}^{H}_{v_{i,n_{i}}, 0}(\mathbf{v},s) = \tilde{\theta}^{\mathcal{M}}_{h_{i, 1}, 0}(\mathbf{u},s) \Bigg|_{\mathbf{u}=\mathbf{v}}.$$
By using the relation $\partial F_H^o / \partial v_{i, j} = \tilde{\theta}^H_{v_{i, j}, 0}$, the first derivatives of the open potential with respect to the rescaled coordinates $\hat{v}_{i, p} = n_i v_{i, n_i - p}$ satisfy the following stability conditions in the limit $n \to \infty$:
$$\lim_{n \to \infty} \frac{\partial F_H^o}{\partial \hat{v}_{0, p}} = (p - 1)! \left. \left( \tilde{\theta}_{s, p - 1}^{\mathcal{M}}(\mathbf{u}) + \tilde{\theta}_{e, p - 1}^{\mathcal{M}}(\mathbf{u}) \right) \right|_{\mathbf{u} = \mathbf{v}};$$
$$\lim_{n \to \infty} \frac{\partial F_H^o}{\partial \hat{v}_{k, p}} = (p - 1)! \left. \tilde{\theta}_{h_{k, 0}, p - 1}^{\mathcal{M}}(\mathbf{u}) \right|_{\mathbf{u} = \mathbf{v}};$$
$$\lim_{n \to \infty} \frac{\partial F_H^o}{\partial \hat{v}_{k, 0}} = \left. \tilde{\theta}_{h_{k, 1}, 0}^{\mathcal{M}}(\mathbf{u}) \right|_{\mathbf{u} = \mathbf{v}}.$$

These identities establish the identification of the hierarchy arising from the stability limit with the open extension of the Whitham hierarchy.
\begin{cor}\label{opencor}
Consider the system for the pair $(\mathcal{F}, \mathcal{F}^o)$, where $\mathcal{F}$ satisfies the equations in Corollary \ref{stabhier} and $\mathcal{F}^o$ is defined by the following equations for $p\ge 1$:
$$\frac{\partial \mathcal{F}^o}{\partial \sigma_{0, p}} = p \left( \lim_{n \to \infty} \frac{\partial F_H^o}{\partial \hat{v}_{0, p}} \right);$$
$$\frac{\partial \mathcal{F}^o}{\partial \sigma_{k, p}} = p  \left( \lim_{n \to \infty} \frac{\partial F^o_{H}}{\partial \hat{v}_{k, p}} \right); $$
$$\frac{\partial \mathcal{F}^o}{\partial \sigma_{k, 0}} = \left( \lim_{n \to \infty} \frac{\partial F^o_{H}}{\partial \hat{v}_{k, 0}} \right)$$
where the variables on the right-hand side are evaluated under the substitutions:
$$\hat{v}_{k, r} = \frac{1}{r} \partial_x \left( \frac{\partial \mathcal{F}}{\partial \sigma_{k, r}} \right), \quad s = \partial_x \mathcal{F}^o.$$
Consequently, this system coincides with the open $\tau$-structure of the Whitham hierarchy via the following identification:
$$\frac{\partial}{\partial \sigma_{0, p}} = p! \left( \frac{\partial}{\partial \widetilde{T}^{e, p-1}} + \frac{\partial}{\partial \widetilde{T}^{s, p-1}} \right);$$
$$\frac{\partial}{\partial \sigma_{k, p}} = p! \frac{\partial}{\partial \widetilde{T}^{h_{k, 0}, p-1}};$$
$$\frac{\partial}{\partial \sigma_{k, 0}} = \frac{\partial}{\partial \widetilde{T}^{h_{k, 1}, 0}}.$$
\end{cor}

  We remark that in Corollary \ref{opencor}, the spatial derivative $\partial_x$ is identified with $\partial / \partial \widetilde{T}^{e, 0}$. However, this flow is absent from the system spanned by the $\sigma$-flows, rendering the $\sigma$-description of the open $\tau$-structure incomplete. This structural deficiency is rooted in the instability of the derivative $\partial F_H^o / \partial s$ in the $n \to \infty$ limit, which corresponds to the missing flows along the open sectors $\partial / \partial \widetilde{T}^{s, p}$.

Upon specializing to $m = 0$, our construction reduces to the hierarchy associated with the type $A$ singularity established in \cite{basalaev2021open}. While the Hirota bilinear equations for this hierarchy were previously derived in that work, our universal construction provides an intrinsic Lax representation. 

Analogous results hold for the alternative solution to the open WDVV equations constructed in \cite{ma2025solutions}, up to linear combinations of the flows. For the sake of brevity, these derivations are not explicitly presented here.

\section{Stabilization of D-type prepotentials}
We now consider the reduction of the foregoing results to the submanifold $ H_{0, \mathbf{n}'}^{\mathrm{even}} $, which, in certain cases, coincides with the Frobenius manifold associated with the D-type singularity \cite{zuo2007frobenius}. Assume there exist positive integers $ m' $ and $\mathbf{n}' = (n_0', n_1', \dots, n_{m'}' )$ such that the following relations hold:
$$
m=2m'-1,\quad n_{0}=2n_{0}',\quad n_{1}=2n_{1}',\quad n_{2j-2}= n_{2j-1}=n_{j}',\quad 2\le j\le m'.
$$
Let $H_{0, \mathbf{n}'}^{\mathrm{even}}$ be the submanifold of $H_{0, \mathbf{n}}$ consisting of superpotentials of the form:
\begin{equation}\label{Dckpsup}
	\lambda(z)=z^{2n_{0}'}+\sum_{j=0}^{n_{0}'-1}b_{0,j}z^{2j}+\sum_{j=1}^{n_{1}'}b_{1,j}z^{-2j}+\sum_{i=2}^{m'}\sum_{j=1}^{n_{i}'}b_{i,j}(z^2-b_{i,0})^{-j}.
\end{equation}
This submanifold is defined by the parity symmetry $\lambda(z) = \lambda(-z)$. It is established that $H_{0, \mathbf{n}'}^{\mathrm{even}}$ inherits the Frobenius manifold structure from $H_{0, \mathbf{n}}$ \cite{ma2023principal}.

In terms of flat coordinates, the parity symmetry induces the following constraints on $ H_{0, \mathbf{n}} $:
\begin{align*}
	&v_{0,2} = v_{0,4} = \cdots = v_{0,2n_{0}'-2} = 0; \\
	&v_{1,0} = v_{1,2} = \cdots = v_{1,2n_{1}'} = 0; \\
	&v_{2i-2,j} = -v_{2i-1,j}, \quad  2 \leq i \leq m', \ 0 \leq j \leq n_{i}'.
\end{align*}
Consequently, the flat coordinates for $H_{0, \mathbf{n}'}^{\mathrm{even}}$ are identified as:
$$
\mathbf{v}^{even} = \{v_{0,2j-1}\}_{j=1}^{n_{0}'} \cup \{v_{1,2j-1}\}_{j=1}^{n_{1}'} \cup \{v_{2,j}\}_{j=0}^{n_{2}'} \cup \{v_{4,j}\}_{j=0}^{n_{3}'} \cup \cdots \cup \{v_{2m'-2,j}\}_{j=0}^{n_{m'}'}.
$$

Restricting the densities from $H_{0, \mathbf{n}}$ to the submanifold $H_{0, \mathbf{n}'}^{\mathrm{even}}$ yields
\begin{align*}
	&\theta_{v_{0,2j-1},p}^{even}=\theta_{v_{0,2j-1},p},\quad j=1,2,\cdots,n_{0}';\\
	&\theta_{v_{1,2j-1},p}^{even}=\theta_{v_{1,2j-1},p},\quad j=1,2,\cdots,n_{1}';
\end{align*}
and	
$$
\theta_{v_{2i-2,j},p}^{even}=\theta_{v_{2i-2,j},p}-\theta_{v_{2i-1,j},p},\quad i=2,\cdots,m',\ j=0,\cdots,n_{i}'. 
$$
The restricted flows then constitute the principal hierarchy for $H_{0, \mathbf{n}'}^{\mathrm{even}}$. The associated $\tau$-structure for the principal hierarchy is then characterized by
$$
\Omega^{even}_{v_{0,2j-1},0;v_{1,2k-1},0}=\Omega_{v_{0,2j-1},0;v_{1,2k-1},0};
$$
$$
\Omega^{even}_{v_{2i-2,j},0;v_{0,2k-1},0}=\Omega_{v_{2i-2,j},0;v_{0,2k-1},0}-\Omega_{v_{2i-1,j},0;v_{0,2k-1},0}
$$
and 
$$\Omega^{\mathrm{even}}_{v_{2i - 2, j}, 0; v_{2l - 2, k}, 0} = \Omega_{v_{2i - 2, j}, 0; v_{2l - 2, k}, 0} - \Omega_{v_{2i - 1, j}, 0; v_{2l - 2, k}, 0} - \Omega_{v_{2i - 2, j}, 0; v_{2l - 1, k}, 0} + \Omega_{v_{2i - 1, j}, 0; v_{2l - 1, k}, 0}.$$
The prepotential $F_{H^{\mathrm{even}}}$ of the Frobenius manifold $H_{0, \mathbf{n}'}^{\mathrm{even}}$ satisfies
$$
\frac{\p^{2}F_{H^{even}}}{\p v_{i,j}\p v_{k,l}}=\Omega^{even}_{v_{i,j},0;v_{k,l},0},\quad v_{i,j},v_{k,l}\in\mathbf{v}^{even}.
$$

We now pass to the infinite-dimensional setting. Let $\mathcal{M}^{\mathrm{even}} \subset \mathcal{M}^{\mathrm{formal}}$ be the submanifold defined by the following symmetry constraints:
$$-\lambda_{0}(z) = \lambda_{0}(-z), \quad -\lambda_{1}(z) = \lambda_{1}(-z), \quad \lambda_{2j-2}(-z) = \lambda_{2j-1}(z), \quad j = 2, 3, \dots, m'.$$
Subject to these symmetries, the Hamiltonian densities
$$\theta_{e, 2p}, \quad \theta_{h_{1, 0}, 2p}, \quad \theta_{h_{2k - 2, 0}, p} - \theta_{h_{2k - 1, 0}, p}, \quad k = 2, \dots, m'$$
and the corresponding flows of the principal hierarchy
$$\frac{\partial}{\partial T^{e, 2p}}, \quad \frac{\partial}{\partial T^{h_{1, 0}, 2p}}, \quad \frac{\partial}{\partial T^{h_{2k - 2, 0}, p}} - \frac{\partial}{\partial T^{h_{2k - 1, 0}, p}}, \quad k = 2, \dots, m'$$
descend to a well-defined restriction on $\mathcal{M}^{\mathrm{even}}$. Furthermore, the induced $ \tau $-structure on $\mathcal{M}^{\mathrm{even}}$ is prescribed by the following reduction rules: 
$$\Omega^{\mathrm{even}}_{\alpha, 2p; \beta, 2q} = \Omega_{\alpha, 2p; \beta, 2q},\quad \alpha, \beta \in \{e, h_{1, 0}\};$$
$$\Omega^{\mathrm{even}}_{\alpha, 2p; h_{k, 0}, q} = \Omega_{\alpha, 2p; h_{2k - 2, 0}, q} - \Omega_{\alpha, 2p; h_{2k - 1, 0}, q},\quad \alpha \in \{e, h_{1, 0}\};$$
$$\Omega^{\mathrm{even}}_{h_{k, 0}, p; h_{l, 0}, q} = \Omega_{h_{2k - 2, 0}, p; h_{2l - 2, 0}, q} - \Omega_{h_{2k - 2, 0}, p; h_{2l - 1, 0}, q} - \Omega_{h_{2k - 1, 0}, p; h_{2l - 2, 0}, q} + \Omega_{h_{2k - 1, 0}, p; h_{2l - 1, 0}, q}.$$
for $k, l \in \{2, \dots, m'\}.$ 

This reduced principal hierarchy corresponds to the Frobenius manifold structure on
 $\hat{\mathcal{M}}^{\mathrm{even}}$ \cite{wu2012class,ma2025solutions},
consisting of Laurent series of the form:
$$\hat{\lambda}_{0}(z) = z^2 + a_{0} + a_{2}z^{-2} + \dots, \quad \hat{\lambda}_{1}(z) = \hat{a}_{1,-2}z^{-2} + \hat{a}_{1,0} + \dots,$$
$$\hat{\lambda}_{2j-2}(z) = \hat{\lambda}_{2j-1}(-z) = \hat{a}_{j,-1}(z - \varphi_{j})^{-1} + \hat{a}_{j,0} + \hat{a}_{j,1}(z - \varphi_{j}) + \dots, \quad j = 1, \dots, m'.$$
The isomorphism from $\mathcal{M}^{\mathrm{even}}$ to $\hat{\mathcal{M}}^{\mathrm{even}}$ is defined via the identification:$$\hat{\lambda}_{0}(z) = \lambda_{0}(z)^{2}, \quad \hat{\lambda}_{1}(z) = \lambda_{1}(z)^{2}, \quad \hat{\lambda}_{j}(z) = \lambda_{j}(z).$$
For $ m' = 1 $, this principal hierarchy coincides with the dispersionless two-component BKP hierarchy \cite{wu2012class}.

By analogy with Theorem \ref{preptau} and Corollary \ref{stabhier}, the system of equations defined via the stability limit for $F_{H^{even}}$ is given as follows:
\begin{cor}\label{evenstab}
	The hierarchy governed by the following equations is well-defined:
		$$\frac{\partial^{2} \mathcal{F}^{\mathrm{even}}}{\partial \sigma_{i, p} \partial \sigma_{j, q}} = p q \left. \left( \lim_{n \to \infty} \frac{\partial^{2} F_{H^{\mathrm{even}}}}{\partial \hat{v}_{i, p} \partial \hat{v}_{j, q}} \right) \right|_{\hat{v}_{k, r} = \frac{1}{r} \frac{\partial^{2} \mathcal{F}^{\mathrm{even}}}{\partial \sigma_{0, 1} \partial \sigma_{k, r}}};$$
		$$\frac{\partial^{2} \mathcal{F}^{\mathrm{even}}}{\partial \sigma_{i, p} \partial \sigma_{j, 0}} = p \left. \left( \lim_{n \to \infty} \frac{\partial^{2} F_{H^{\mathrm{even}}}}{\partial \hat{v}_{i, p} \partial \hat{v}_{j, 0}} \right) \right|_{\hat{v}_{k, 0} = \frac{\partial^{2} \mathcal{F}^{\mathrm{even}}}{\partial \sigma_{0, 1} \partial \sigma_{k, 0}}};$$
	$$\frac{\partial^{2} \mathcal{F}^{\mathrm{even}}}{\partial \sigma_{i, 0} \partial \sigma_{j, 0}} = \left. \left( \lim_{n \to \infty} \frac{\partial^{2} F_{H^{\mathrm{even}}}}{\partial \hat{v}_{i, 0} \partial \hat{v}_{j, 0}} \right) \right|_{\hat{v}_{k, 0} = \frac{\partial^{2} \mathcal{F}^{\mathrm{even}}}{\partial \sigma_{0, 1} \partial \sigma_{k, 0}}},$$
	where the index of $\sigma$ belongs to the set $\{ (0, 2k - 1) \}_{k \ge 1} \cup \{ (1, 2k - 1) \}_{k \ge 1} \cup \{ (2k - 2, p) \}_{k = 2, \ldots, m', p \ge 0}$.
	
Furthermore, this system constitutes a $\tau$-cover of the principal hierarchy for
	$\mathcal{M}^{even}$, with the flows identified as follows:
$$\frac{\partial}{\partial \sigma_{0, p}} = p!  \frac{\partial}{\partial T^{e, p-1}},$$
$$\frac{\partial}{\partial \sigma_{1, p}} = p! \frac{\partial}{\partial T^{h_{1, 0}, p-1}},$$
$$\frac{\partial}{\partial \sigma_{i, p}} = p! \left( \frac{\partial}{\partial T^{h_{i, 0}, p-1}} - \frac{\partial}{\partial T^{h_{i+1, 0}, p-1}} \right),$$
$$\frac{\partial}{\partial \sigma_{i, 0}} = \frac{\partial}{\partial T^{h_{i, 1}, 0}} - \frac{\partial}{\partial T^{h_{i+1, 1}, 0}}$$
for $i \in \{2, 4, \ldots, 2m' - 2\}$.
\end{cor}
For $ m' = 1 $, we recover the results established in \cite{basalaev2021integrable,basalaev2024b}.

The reduction scheme extends naturally to the open extension. Following the results in \cite{ma2025solutions,ma2023principal}, the solution $F_{H}^{o}$ to the open WDVV equation for $ H_{0, \mathbf{n}} $, together with its associated principal hierarchy, admits a consistent restriction to $H_{0, \mathbf{n}'}^{\mathrm{even}}$.
 The function $\tilde{\theta}^{H^{even}}_{v_{i,j},0}$ associated to open potential $F_{H^{\mathrm{even}}}^{o}$ for $H_{0, \mathbf{n}'}^{\mathrm{even}}$ are given by 
$$\tilde{\theta}^{H^{\mathrm{even}}}_{v_{k, 2j - 1}, 0} = \tilde{\theta}^{H}_{v_{k, 2j - 1}, 0},\quad k=0,1;$$
$$ \tilde{\theta}^{H^{\mathrm{even}}}_{v_{2i - 2, j}, 0} = \tilde{\theta}^{H}_{v_{2i - 2, j}, 0} - \tilde{\theta}^{H}_{v_{2i - 1, j}, 0},\quad i=2,\cdots,m';$$
and
$$\tilde{\theta}^{H^{\mathrm{even}}}_{s, 0} = \tilde{\theta}^{H}_{s, 0}.$$

In the infinite-dimensional setting, as demonstrated in \cite{ma2025solutions}, there exists a pair of solutions to the open WDVV equations  for $\hat{\mathcal{M}}^{even}$. The functions $\tilde{\theta}^{\hat{\mathcal{M}}^{even}}_{\alpha,p}$ associated with one of these solutions, $ F^{o}_{\hat{\mathcal{M}}^{\mathrm{even}}} $, are given by:
$$\tilde{\theta}^{\hat{\mathcal{M}}^{\mathrm{even}}}_{e, 2p} = \tilde{\theta}^{\mathcal{M}}_{e, 2p}, \quad \tilde{\theta}^{\hat{\mathcal{M}}^{\mathrm{even}}}_{h_{1, 0}, 2p} = \tilde{\theta}^{\mathcal{M}}_{h_{1, 0}, 2p},$$
$$\tilde{\theta}^{\hat{\mathcal{M}}^{\mathrm{even}}}_{h_{2k-2, 0}, p} = \tilde{\theta}^{\mathcal{M}}_{h_{2k - 2, 0}, p} - \tilde{\theta}^{\mathcal{M}}_{h_{2k - 1, 0}, p},$$
$$\tilde{\theta}^{\hat{\mathcal{M}}^{\mathrm{even}}}_{h_{2k, 1}, 0} = \tilde{\theta}^{\mathcal{M}}_{h_{2k - 2, 1}, 0} - \tilde{\theta}^{\mathcal{M}}_{h_{2k - 1, 1}, 0},$$
for $k = 2, \dots, m', p \ge 0,$ and
$$\tilde{\theta}^{\hat{\mathcal{M}}^{\mathrm{even}}}_{s, p} = \tilde{\theta}^{\mathcal{M}}_{s, p}.$$

By analogy with Corollary \ref{opencor}, we obtain a hierarchy from the stability limit of the open potential $F_{H^{\mathrm{even}}}^{o}$.
\begin{cor}
	Consider the system for the pair $(\mathcal{F}^{even}, \mathcal{F}^{o,even})$, where $\mathcal{F}^{\mathrm{even}}$ satisfies the equations in Corollary \ref{evenstab} and $\mathcal{F}^{o, \mathrm{even}}$ is governed by the following system:
	
  $$\frac{\partial \mathcal{F}^{o, \mathrm{even}}}{\partial \sigma_{0, p}} = p \left( \lim_{n \to \infty} \frac{\partial F^{o}_{H^{\mathrm{even}}}}{\partial \hat{v}_{0, p}} \right);$$
  $$\frac{\partial \mathcal{F}^{o, \mathrm{even}}}{\partial \sigma_{k, p}} = p \left( \lim_{n \to \infty} \frac{\partial F^{o}_{H^{\mathrm{even}}}}{\partial \hat{v}_{k, p}} \right);$$
 $$\frac{\partial \mathcal{F}^{o, \mathrm{even}}}{\partial \sigma_{k, 0}} = \left( \lim_{n \to \infty} \frac{\partial F^{o}_{H^{\mathrm{even}}}}{\partial \hat{v}_{k, 0}} \right).$$
  In these equations, the variables on the right-hand side are subject to the following substitutions:
   $$\hat{v}_{k, r} = \frac{1}{r} \frac{\partial^{2} \mathcal{F}^{\mathrm{even}}}{\partial x \partial \sigma_{k, r}}, \quad s = \frac{\partial \mathcal{F}^{o, \mathrm{even}}}{\partial x}.
   $$
   
Consequently, this system is identified with the open $ \tau $-cover of the reduced hierarchy on the loop space of $ \mathbf{M}^{\mathrm{even}} \times \mathbf{C} $ as follows:
	$$\frac{\partial}{\partial \sigma_{0, p}} = p! \left( \frac{\partial}{\partial \widetilde{T}^{e, p-1}} + \frac{\partial}{\partial \widetilde{T}^{s, p-1}} \right);$$
	$$\frac{\partial}{\partial \sigma_{1, p}} = p! \frac{\partial}{\partial \widetilde{T}^{h_{1, 0}, p-1}};$$
	$$\frac{\partial}{\partial \sigma_{k, p}} = p! (\frac{\partial}{\partial \widetilde{T}^{h_{k, 0}, p-1}}-\frac{\partial}{\partial \widetilde{T}^{h_{k+1, 0}, p-1}});$$
	$$\frac{\partial}{\partial \sigma_{k, 0}} = \frac{\partial}{\partial \widetilde{T}^{h_{k, 1}, 0}}-\frac{\partial}{\partial \widetilde{T}^{h_{k+1, 1}, 0}}$$
	for $k \in \{2, 4, \dots, 2m'-2\}.$
\end{cor}
In the special case $ m' = 1 $, our construction recovers the results presented in \cite{basalaev2022integrable}.


\end{document}